\newif\ifprocs
\procstrue
\procsfalse 

\newif\ifarxiv
\arxivtrue
\arxivfalse 

\newif\ifcomments
\commentstrue
\commentsfalse 

\documentclass[11pt,USenglish]{article}
\usepackage[margin=1.0in]{geometry}
\usepackage[T1]{fontenc}
\usepackage[utf8]{inputenc}
\usepackage[english]{babel}
\usepackage{authblk}
\usepackage{chngcntr}
\usepackage[T1]{fontenc}
\counterwithin{equation}{section}
\usepackage{graphicx}
\usepackage{enumitem}
\usepackage{mwe}
\usepackage{makecell}
\usepackage{caption}
\usepackage{xcolor}
\usepackage{algorithm}
\usepackage[noend]{algpseudocode}

\usepackage{amsmath}
\usepackage{amssymb}
\usepackage{complexity}
\usepackage{graphicx}
\usepackage[export]{adjustbox}
\usepackage{thmtools}
\usepackage{amsthm}
\usepackage{thm-restate}
\usepackage{xspace}
\usepackage[colorlinks,linkcolor=blue,citecolor=blue,filecolor=blue,urlcolor=blue]{hyperref}

\newtheorem{theorem}{Theorem}[section]
\newtheorem{lemma}[theorem]{Lemma}

\newtheorem{definition}[theorem]{Definition}

\newtheorem{corollary}[theorem]{Corollary}

\newtheorem{infthm}[theorem]{Informal Theorem}

\theoremstyle{plain}
\newtheorem{claim}[theorem]{Claim}

\makeatletter
\newtheorem*{rep@theorem}{\rep@title}
\newcommand{\newreptheorem}[2]{
\newenvironment{rep#1}[1]{
 \def\rep@title{#2 \ref{##1}}
 \begin{rep@theorem}}
 {\end{rep@theorem}}}
\makeatother
\newreptheorem{theorem}{Theorem}

\makeatletter
\newtheorem*{rep@corollary}{\rep@title}
\newcommand{\newrepcorollary}[2]{
\newenvironment{rep#1}[1]{
 \def\rep@title{#2 \ref{##1}}
 \begin{rep@corollary}}
 {\end{rep@corollary}}}
\makeatother
\newreptheorem{corollary}{Corollary}
\newreptheorem{lemma}{Lemma}

\def\compactify{\itemsep=0pt \topsep=0pt \partopsep=0pt \parsep=0pt}

\newcommand{\ProblemName}[1]{\textsf{#1}}
\newcommand{\MF}{\ProblemName{Max-Flow}\xspace}

\newcommand{\APMF}{\ProblemName{All-Pairs Max-Flow}\xspace}

\newcommand{\STMF}{\ProblemName{ST-Max-Flow}\xspace}

\newcommand{\MFV}{\ProblemName{Max-Flow}\xspace}
\newcommand{\SETH}{\ProblemName{SETH}\xspace}
\newcommand{\NSETH}{\ProblemName{NSETH}\xspace}
\newcommand{\NUNSETH}{\ProblemName{NUNSETH}\xspace}

\newcommand{\TOV}{\ProblemName{$3$OV}\xspace}

\newcommand{\MCDS}{\ProblemName{Min-Cut data structure}\xspace} 
\newcommand{\MCDSs}{\ProblemName{Min-Cut data structures}\xspace} 
\newcommand{\kSAT}{\ProblemName{$k$-SAT}\xspace}

\newcommand{\MC}{\ProblemName{Min-Cut}\xspace}

\DeclareMathOperator{\SOL}{SOL}

\DeclareMathOperator{\vol}{\textbf{volume}}
\DeclareMathOperator{\val}{\textbf{value}}

\newcommand\eps{\varepsilon}
\renewcommand\epsilon{\varepsilon}
\newcommand\tO{\ensuremath{\tilde O}}
\newcommand\hO{\ensuremath{\hat O}}

\newcommand{\Q}{\emph{\textbf{Q1}}\xspace}
\newcommand{\QQ}{\emph{\textbf{Q2}}\xspace}
\newcommand{\TPM}{\tilde{t}_{p}(m)\xspace}
\newcommand{\TMCM}{\tilde{t}_{mc}(m)\xspace}
\newcommand{\TMC}{\tilde{t}_{mc}\xspace}

\providecommand{\set}[1]{{\{#1\}}}
\providecommand{\card}[1]{\lvert#1\rvert}
\usepackage{lipsum}
\ifarxiv\else
\graphicspath{{Figures/}}
\fi 

\begin{document}

\title{
(Almost) Ruling Out SETH Lower Bounds for All-Pairs Max-Flow}

\author[1]{Ohad Trabelsi\thanks{Partially supported by the  NWO VICI grant 639.023.812.}
}
\affil[1]{Toyota Technological Institute at Chicago. Email: \texttt{ohadt@ttic.edu}
}

\maketitle

\begin{abstract}

The All-Pairs Max-Flow problem has gained significant popularity in the last two decades, and many results are known regarding its fine-grained complexity. Despite this, wide gaps remain in our understanding of the time complexity for several basic variants of the problem, including for directed or undirected input graphs that are edge- or node-capacitated, and where the capacities are unit or arbitrary.
In this paper, we aim to bridge these gaps by providing algorithms, conditional lower bounds, and non-reducibility results.
Notably, we show that for most problem settings, deterministic reductions based on the Strong Exponential Time Hypothesis (SETH) cannot rule out $O(n^{4-\varepsilon})$ time algorithms for some small constant $\varepsilon>0$, under a hypothesis called NSETH.

To obtain our results for undirected graphs with unit node-capacities (aka All-Pairs Vertex Connectivity), we design a new \emph{randomized Las Vegas} $O\big(m^{2+o(1)}\big)$ time combinatorial algorithm. This is our main technical result, improving over the recent $O\big(m^{11/5+o(1)}\big)$ time \textit{Monte Carlo} algorithm [Huang et al., STOC 2023] and matching their $m^{2-o(1)}$ lower bound (up to subpolynomial factors), thus essentially settling the time complexity for this setting of the problem.
\end{abstract}

\newpage
\section{Introduction}
\label{sec:intro}

The maximum $st$-flow problem, denoted \MF, has been extensively studied in theoretical computer science for almost 70 years.
This problem has become a key tool in algorithm design, 
with numerous applications of the problem and its techniques~\cite{AMO93,FF56,AHK12,WL93}. 
In a recent breakthrough,~\cite{CKLP22,BCG+23} gave an almost linear $\hO(m)$ time algorithm for the problem~\footnote{Henceforth, $n$ denotes the number of nodes, and $m$ the number of edges in an input graph.}\textsuperscript{,}\footnote{$\hO(\cdot)$ hides $n^{o(1)}$ factors and $\tO(\cdot)$ hides $\polylog(n)$ factors.}\textsuperscript{,}\footnote{We say that a time bound $T(m)$ is almost linear if it is bounded by $m^{1+o(1)}$}.
But what if we desire to find the \MF value for all pairs? This problem is named \APMF, and the naive approach for it always works: simply apply a \MF routine for every pair of nodes, thus resulting in an $n^2 \cdot m^{1+o(1)}$ time algorithm. 
Then the following question arises:
\begin{center}
\emph{\textbf{Q1}: Can All-Pairs Max-Flow be solved faster than the naive approach? \\Alternatively, can we show conditional lower bounds?}
\par\end{center}
In the main variants of \APMF considered in the literature, the input graphs may be:
\begin{itemize}
\item directed or undirected,
\item edge-capacitated or node-capacitated, and 
\item with unit-capacities, also called \emph{connectivities}, or with arbitrary (usually polynomially-bounded) capacities.
\end{itemize}
It turns out that undirected graphs with edge capacities are particularly easy to solve, as a recent exciting line of work~\cite{LP20,LP21,AKT21_stoc,AKT21,LPS21} on Gomory-Hu tree\footnote{This is a data structure that succinctly represents all min $st$-cuts in an input undirected graph with edge capacities.} construction culminated with an almost linear $\hO(m)$ time algorithm~\cite{AKLPST22,ALPS23}.
However, the remaining variants are unlikely to be solved in almost linear time since hardness results are known for them (see Table~\ref{table:bounds}), and so henceforth we call them \emph{the hard variants}.
In this context, the Strong Exponential Time Hypothesis (\SETH)\footnote{\SETH postulates that for every $\varepsilon>0$ there exists an integer $k>0$ such that \kSAT requires $\Omega(2^{(1-\varepsilon)n})$ time.}~\cite{Impa01spar} has been particularly useful, proving $n^{3-o(1)}$ time conditional lower bounds for hard variants in sparse $m=\tO(n)$ graphs with (node or edge) capacities~\cite{KT18,AKT20}, and an $n^{3-o(1)}$ conditional lower bound for dense $m=O(n^2)$ graphs with unit edge-capacities.

Despite the lower bounds mentioned, the answer to \Q is positive in some restricted cases. For unit capacities in hard variants where either the input graph is sparse or the desired flow value is bounded, nontrivial algorithms are known. These include a network-coding based $\tO(m^{\omega})$ time algorithm for directed graphs with unit edge-capacities~\cite{CLL13}, an $(nk)^{\omega}$ time adaptation of the latter to the $k$-bounded version of unit node-capacities (i.e., where the flow value is bounded by $k$), and an $\hO(mk^3)$ time algorithm for the $k$-bounded version of undirected graphs with unit node-capacities. The latter was obtained through an adaptation of an approximate Gomory-Hu tree algorithm~\cite{LP21} and employing an isolating cut lemma~\cite{LP20,AKT21_stoc,LNPSY21}. Very recently, via a high-degree low-degree method previously used in the context of Gomory-Hu trees~\cite{AKT20}, the $\hO(mk^3)$ algorithm above was adapted to an $\hO(m^{11/5})$ time algorithm for the unbounded version~\cite{HLSW23}.
Still, in spite of extensive research, the naive approach remains the fastest algorithm for general dense graphs in all hard variants, while no matching conditional lower bounds are known.
A relaxed phrasing of \Q restricts attention to combinatorial algorithms.\footnote{\emph{combinatorial algorithms} commonly refers to algorithms avoiding fast matrix multiplication as a subroutine.}
\newpage
\begin{center}
\emph{\textbf{Q2}: What is the time complexity of All-Pairs Max-Flow when restricted to combinatorial algorithms?}
\par\end{center}

A typical hope is that this restriction would facilitate better conditional lower bounds.
Indeed, at least for dense graphs, \QQ is settled. By a reduction from the problem of detecting a $4$-clique in a graph, it was shown~\cite{A+18} that all directed variants are unlikely to admit $O(n^{4-\varepsilon})$ time combinatorial algorithms.
Later, this technique was extended to undirected graphs with unit node-capacities~\cite{HLSW23}, implying an $n^{4-o(1)}$ conditional lower bound against combinatorial algorithms for them as well.
This result finally established that \emph{no} hard variant can be solved faster than the naive approach by a combinatorial algorithm, at least when the input graph is dense.
Notably, all these $n^{4-o(1)}$ conditional lower bounds can also be expressed as $m^{2-o(1)}$ conditional lower bounds for all $m$ (accomplished by simply adding isolated nodes).
This still leaves a gap for less dense graphs in various settings, including undirected graphs with unit node-capacities, for which the current upper bound is $\min\{\hO(n^2\cdot m), \hO(m^{11/5})\}$ (see Figure~\ref{Figs:Algs}).
\subsection{A New Algorithm for Graphs with Unit Node-Capacities}

In our main result we resolve \textbf{\emph{Q2}} for this setting across all densities by closing the gap with the known conditional lower bound against combinatorial algorithms (see also Figure~\ref{Figs:Algs}), settling the time complexity of \APMF in these settings.

\begin{theorem}\label{thm:algorithm:main}
There is a randomized Las Vegas algorithm that given an undirected graph $G=(V,E)$ with unit node-capacities, makes $\tO(m)$ calls to \MF in $G$ and spends $\tO(mn)$ time outside of these calls, and can:
\begin{itemize}
\item construct a data structure of size $\tO\left(m^{3/2}\right)$ that stores a set $\mathcal{C}$ of $\tO(m)$ vertex-cut separators\footnote{For a pair $u,v\in V$, a vertex-cut $uv$-separator is a set of nodes separating $u$ and $v$.}, such that given a queried pair $u,v\in V$ outputs in $\tO(1)$ time a pointer to an optimal $uv$-separator in $\mathcal{C}$.
\item
construct a data structure of size $\tO(m)$ that given a queried pair $u,v\in V$ outputs in \\$\tO(\MF(u,v))$ time\footnote{\MF values greater than $\sqrt{m}$ can be output in $\tO(1)$ time using high-degree considerations.} the value $\MF(u,v)$
(or simply in $\tO(1)$ time in an alternative Monte Carlo version).
\end{itemize}
\end{theorem}

To obtain this bound, we depart from the recent practice of using element-connectivity\footnote{Given a graph $G=(V,E)$ and a terminal subset $T\subseteq V$, the element-connectivity between $u,v\in T$ is the maximum number of paths between $u$ and $v$ in $G$ that are disjoint with respect to $E$ and also to $V\setminus T$.}
for constructing vertex-connectivity data structures~\cite{PSY22,HLSW23}, and adapt an algorithm previously employed for $(1+\varepsilon)$-approximate \APMF on undirected, edge-capacitated graphs~\cite{AKT20}.
Note that our algorithm is Las Vegas, unlike many of the recent algorithms for \APMF which are Monte Carlo~\cite{PSY22,AKLPST22,HLSW23,ALPS23}.
This algorithm is also combinatorial, and runs in $\hO(m^2)$ time using the recent near linear $\hO(m)$ time \MF algorithm~\cite{CKLP22,BCG+23}, improving on the recent $\hO(m^{11/5})$~\cite{HLSW23} time algorithm for the problem.

\begin{figure}[t]
       \includegraphics[width=0.8\textwidth,left]{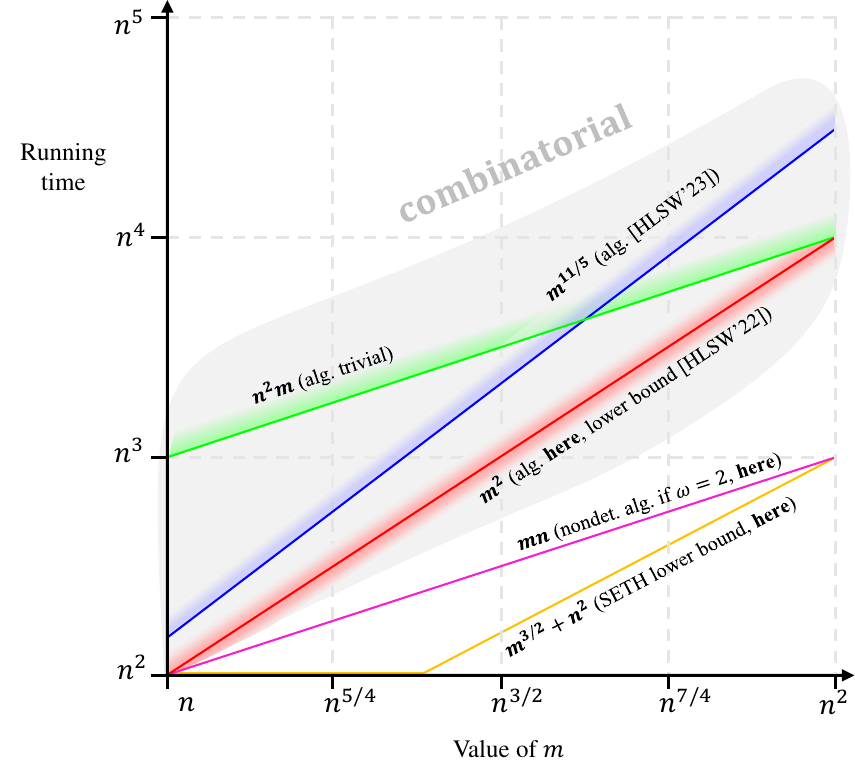}
   \caption[-]{
New and old bounds for \APMF on undirected graphs with unit node-capacities.
   }
   \label{Figs:Algs}
\end{figure}

\subsection{Other Results}
Following our resolution of \QQ for unit node-capacities, the main open question remains: can we answer \Q by proving $n^{4-o(1)}$ conditional lower bounds for general algorithms too? 
Perhaps surprisingly, we establish a strong barrier to the possibility of such a result.
We address \textbf{\emph{Q1}} by showing that for most settings of \APMF, deterministic \SETH-based reductions proving $n^{4-o(1)}$ conditional lower bounds can be ruled out assuming a hypothesis called \NSETH\footnote{The Nondeterministic Strong Exponential Time Hypothesis (\NSETH) asserts that \SETH holds against co-nondeterministic algorithms.}~\cite{CGIMPS16}.
We do so by designing subquartic nondeterministic algorithms\footnote{A nondeterministic algorithm can make guesses (i.e. nondeterministic choices) and eventually produce either a correct output, or ``don't know'' (i.e. aborts). Importantly, every input has at least one sequence of guesses leading to a correct output.} suited for each setting, and then applying a non-reducibility framework by Carmosino et al.~\cite{CGIMPS16}.
First, we present a general method by which we obtain our nondeterministic algorithms.
In what follows, the volume of an $st$-flow $f_{st}$, denoted $\vol(f_{st})$, is the number of edges in the support of $f_{st}$ in the input graph.

\begin{infthm}\label{thm:informal_nondet}
Suppose there is a randomized Las Vegas algorithm for \APMF that makes $q$ \MF calls on input graphs $G=(V,E)$ and spends $\hO(t)$ time outside of these calls, and let $\mu$ be a volume upper bound on the maximum $st$-flow in $G$, for all pairs $s,t\in V$. Then there is a nondeterministic algorithm for \APMF with time $\hO(q\cdot\mu+n^{\omega(\log_{n}q,1,1)}+t)$\footnote{$\omega(a,b,c)$ is the time it takes to multiply an $n^a\times n^b$ matrix by an $n^b\times n^c$ matrix.}.
\end{infthm}  
Note that for $\omega=2$, this time becomes $\hO(q\cdot\mu+qn+t)$.
In the remaining of this section we present our non-reducibility results for each variant.

\subsubsection{Unit Node-Capacities}

For undirected graphs with unit node-capacities, we obtain the following non-reducibility bound by combining Theorem~\ref{thm:informal_nondet} with the natural $O(n)$ flow volume bound and the algorithm in Theorem~\ref{thm:algorithm:main}.

\begin{theorem}\label{thm:undirected_unit_node}
Under \NSETH, for any fixed $\varepsilon>0$ there is no deterministic fine-grained reduction~\footnote{Similar to existing literature, we rule out black-box Turing reductions. See~\cite{CGIMPS16} and Section~\ref{prems} in this paper for more formal definitions.} proving an $\Omega(n^{\omega(\log_n m,1,1)+\varepsilon})$ \SETH-based lower bound for \APMF on undirected graphs with unit node-capacities.
\end{theorem}
For $\omega=2$, the bound becomes $\Omega(mn^{1+\varepsilon})$.
A question posed by~\cite{HLSW23} is whether conditional lower bounds for general algorithms could match their $m^{2-o(1)}$ lower bounds against combinatorial algorithms.
Theorem~\ref{thm:undirected_unit_node} conditionally answers this question negatively, for \SETH-based results.
That said, we observe a \SETH lower bound for this setting, which is obtained by a straightforward adaptation of the conditional lower bound in~\cite{AKT20_soda} for general node-capacities to the unit node-capacities case, together with a recently developed source-sink isolating gadget~\cite{HLSW23}.
\begin{theorem}\label{thm:CLB}
If for some fixed $\varepsilon>0$ \APMF on undirected graphs with unit node-capacities can be solved in time $O(n^2+m^{3/2-\varepsilon})$ then \SETH is false.
\end{theorem}
\subsubsection{General node-capacities.}

For directed graphs with general node-capacities\footnote{The running times of all our algorithms do not depend on the numerical values of the capacities in the input.}, we obtain the following non-reducibility bound by combining Theorem~\ref{thm:informal_nondet} with the naive algorithm and our new flow volume bounds.
\begin{theorem}\label{thm:nodecap}
Under \NSETH, for any fixed $\varepsilon>0$ there is no deterministic fine-grained reduction proving an $\Omega(n^{\omega(2,1,1)+\varepsilon})$ \SETH-based lower bound for \APMF on directed graphs with node-capacities.
\end{theorem}
For $\omega=2$, the bound becomes $\Omega(n^{3+\varepsilon})$.
\begin{table}[t]
\footnotesize
\centering
\begin{tabular}{c c c c c c c}
\hline
\hline
\makecell{Setting} & Dir. & Capacities & \makecell{SETH\\Lower Bounds} & \makecell{Combinatorial\\Lower Bounds}  & Algorithms & \makecell{Nondeterministic\\Algorithm}
 \\ [0.5ex]
\hline
\hline
node-cap & no & unit & $m^{3/2}+n^2$, \textbf{Thm~\ref{thm:CLB}} & $m^2$~\cite{HLSW23} & $m\cdot T(m)$, \textbf{Thm~\ref{thm:algorithm:main}} & $mn$, \textbf{Lma~\ref{lemma:nondet:undirected_unit_node}} \\
 & yes & unit & $mn$~\cite{KT18} & $m^2$~\cite{A+18} & $m^{\omega}$~\cite{CLL13}, $n^2\cdot T(m)$ & $n^3$, \textbf{Lma~\ref{lemma:nondet:nodecap}} \\
 & no & general & $n^3$~\cite{AKT20_soda} & $m^2$ [ '' ] & $n^2\cdot T(m)$ & $n^3$, \textbf{Lma~\ref{lemma:nondet:nodecap}} \\
 & yes & general & $n^3$ [ '' ] & $m^2$ [ '' ] & $n^2\cdot T(m)$ & $n^3$, \textbf{Lma~\ref{lemma:nondet:nodecap}} \\
[0.5ex]\hline
edge-cap & yes & unit & $mn$ & $m^2$ [ '' ] & $m^{\omega}$~\cite{CLL13}, $n^2\cdot T(m)$ & $n^{5/2}\sqrt{m}$, \textbf{Lma~\ref{lemma:nondet:unit_edge}}\\
 & yes & general & $n^3$~\cite{KT18} & $m^2$ [ '' ] & $n^2\cdot T(m)$ &\\
 & no & general & & & $m^{1+o(1)}$~\cite{ALPS23} &\\
[0.5ex]
\hline
\end{tabular}
\caption{Known algorithms and lower bounds for hard variants of \APMF, ignoring subpolynomial factors, and where $T(m)$ represents the time of a \MF query. All presented bounds for nondeterministic algorithms are under the assumption that $\omega=2$.
}
\label{table:bounds}
\end{table}
This result indicates that even when the input graph is dense and directed, the existing $n^{3-o(1)}$ \SETH lower bound~\cite{AKT20_soda}, designed for sparse undirected graphs, is unlikely to be improved.
\subsubsection{Directed edge-capacities.}
Our last non-reducibility result is for directed graphs with unit edge-capacities.
By combining Theorem~\ref{thm:informal_nondet} with the naive algorithm and a flow volume bound obtained by Karger and Levine~\cite{KL98} in the context of (single pair) \MF algorithms, we show that no deterministic \SETH-based reduction is likely to prove an $n^{4-o(1)}$ conditional lower bound for this setting.

\begin{theorem}\label{thm:unit_edge}
Under \NSETH, for any fixed $\varepsilon,\varepsilon'>0$ there is no deterministic fine-grained reduction proving an 
$\Omega(n^{5/2+\varepsilon} \sqrt{m} + n^{\omega(2,1,1)+\varepsilon'})$ \SETH-based lower bound for \APMF on directed graphs with unit edge-capacities.
\end{theorem}
For $\omega=2$, the bound becomes $\Omega(n^{5/2+\varepsilon} \sqrt{m})$.
Combining our framework with Gusfield's algorithm for Gomory-Hu trees~\cite{Gusfield90} immediately provides an $O(n^{2.5})$ time nondeterministic algorithm for \APMF in undirected graphs with unit edge-capacities.
Although this algorithm may not be competitive given the $\tO(m)$ time nondeterministic algorithm for Gomory-Hu tree construction~\cite{AKT20},
we still find this observation meaningful.
Relatedly, a nondeterministic algorithm was the seed of inspiration for recent breakthroughs in randomized algorithms for this problem~\cite{AKLPST22,ALPS23} (see~\cite{AKT21_stoc,AKT21} for discussions), and we believe that the same could happen to the nondeterministic algorithms presented in this paper.

Finally, for directed graphs with general edge-capacities, we could not prove a subquartic nondeterministic algorithm. However, as a consequence of Theorem~\ref{thm:nodecap}, we conclude that a hypothetical deterministic fine-grained reduction from the edge-capacities setting to the node-capacities setting would imply the non-existence of an $n^{4-o(1)}$ conditional lower bound for the edge-capacities setting.

\begin{theorem}\label{thm:edgecap}
Suppose there exists a deterministic fine-grained reduction proving that an $O(n^{4-\varepsilon})$ time algorithm for \APMF on graphs with node-capacities implies an $O(n^{4-\varepsilon'})$ algorithm for \APMF on directed graphs with edge-capacities, for some fixed $\varepsilon,\varepsilon'>0$.

Then, under \NSETH, it is not possible to have an $n^{4-o(1)}$ \SETH-based lower bound for \APMF on directed graphs with edge-capacities.
\end{theorem}

Whether such a reduction exists is a well-known open problem in the field (see e.g. Open Problem $2$ in~\cite{AKT20_soda}), and Theorem~\ref{thm:edgecap} points to a significant barrier to its positive resolution.

\subsection{Technical Overview}
Below we explain the ideas behind our randomized algorithm for undirected graphs with unit node-capacities. 
Then, we elaborate on our non-reducibility results.

\subsubsection{An algorithm for undirected graph with unit node-capacities}

\paragraph*{High level ideas in the algorithm in~\cite{AKT20}}

One of the core ideas in the $(1+\varepsilon)$-approximation algorithm for graphs with edge-capacities could be interpreted as the following: The minimum cuts between any three nodes can be captured by just two cuts. This is represented also by the following inequality for any set of nodes $u,v,w$,
\begin{align}\label{ineq:triangle}
\MF(u,v) \geq \min\{ \MF(u,w),\MF(v,w)\},
\end{align}
which translates to the algorithmic insight that the minimum cuts between a randomly picked pivot and every other node can capture many more minimum cuts in the graph.
More concretely, it was shown in~\cite{AKT20} that the minimum cuts computed from a random (source) pivot $p$ to all other (sink) nodes $u\in V$ can be used to partition $V$ as follows. 
Let $S_{u,p} \subset V$ be the side of $u$ in a minimum $pu$-cut $(S_{p,u}, S_{u,p})$. Each node $u$ is then reassigned to a set $S_{u',p}$ containing $u$, where nodes reassigned to the same $S_{u',p}$ are grouped together. This partition satisfies that its biggest part is of expected size at most $\frac{3}{4}n$, while preserving property (*) below. This allows the algorithm to continue recursively with an $O(\log n)$ bound on the recursion-tree depth.

\begin{itemize} \compactify
\item[(*)] \label{it:starproperty} 
  For every pair $s' \in S_{i,p}, t'\in S_{j,p}$ for $i \neq j$, at least one of $(S_{i,p}, S_{p,i})$ or $(S_{j,p},S_{p,j})$ corresponds in $G$ to a $(1+\eps)$-approximate minimum $s't'$-cut.
\end{itemize}

Crucially, property (*) ensures that the \MF values between node pairs of different sets are determined at the current recursion step, i.e. one of the minimum cuts computed from $p$ captures the (approximated) value of $\MF(u,v)$, while handling the remaining pairs is postponed to a later stage.

\paragraph*{Barriers in adapting~\cite{AKT20}}

While symmetry holds in our case of undirected graphs with unit node-capacities, the same intuition is not directly applicable due to two main issues.

First, the triangle inequality~\ref{ineq:triangle} does not generally hold. This is illustrated by the example of nodes $p,u,v$ satisfying $v\in C_{p,u}$ or $u\in C_{p,v}$, where $C_{p,v}$ is the separator in a minimum vertex $pv$-cut. 
This makes the presence of nodes in separators of many pairs an issue, as their own \MF to other nodes might remain undetermined for most pivot choices $p$ and through multiple recursion depths, resulting in no cuts for them.

Second, and even if the first issue is resolved, the pivot $p$ itself could be present in optimal vertex $uv$-separators. This means that even when the minimum $pu$-cut separates a pair $v,u$, it might do so in a sub-optimal way that must not include the pivot $p$ in $C_{u,v}$, and we would not have any indication for it. Moreover, there might not be any pivot choice that is absent from all separators, making the situation worse.
To address these issues, we first provide an overview of a Monte Carlo version of our algorithm, to be followed by the necessary changes to achieve a Las Vegas algorithm.

\paragraph*{Overcoming barriers in~\cite{AKT20}: A Monte Carlo algorithm}

Despite the issues mentioned, the underlying intuition is applicable in our setting, and the shared symmetry property turns out to be enough to recover an $\hO(m^2)$ running time.
Our algorithm uses a similar framework as in~\cite{AKT20}, which involves a node partition inspired by minimum cuts from a random pivot and recursion. However, we introduce reinforcements in the form of additional \MF queries to address the issues described. The main challenge is to ensure that the total number of additional queries does not exceed $\tO(m)$, thus maintaining the total time complexity of $\hO(m^2)$ using the almost linear time algorithm for \MF~\cite{CKLP22,BCG+23}, as required.

To address the first issue, we implement the following changes. When a uniformly random node $p\in V'$ is selected, in addition to the \MF calls to all nodes $u\in V'\setminus \{p\}$, we also directly query the \MF between every node $q$ in the separator $C_{p,u}$, and $u$. Recall that in our setting, it holds that $|C_{p,u}|\leq d_G(u)$ where $d_G(u)$ is the degree of node $u$ in the input graph $G$, and therefore the total number of additional queries for $V'$ is at most the sum of node degrees in $V'$. Since the node sets $V'$ in any depth of the recursion are disjoint, the total number of extra queries is $\tO(m)$.

To address the second issue, it is important to note that node pairs with a large separator, that is where both nodes have high degrees, are the most problematic in this context. Therefore, we would like to adopt a similar high degree/low degree approach to the one taken previously in~\cite{AKT20_soda,HLSW23}; solving \MF between all pairs of high degree nodes, and letting the "main" algorithm handle the remaining pairs.

However, there are a few challenges involved in doing so. First, simply querying all pairwise nodes of high degree is not sufficient for our needs. This is because our requirement is exceptionally strong: that a uniformly random pivot in \emph{any} instance $V'$ has a good chance of not being in many vertex separators, even when $|V'|$ is much smaller than $n$, or even $\sqrt{m}$.
To achieve this, we query the \MF between pairs of high degree nodes in every instance, but with a very high degree threshold compared to $|V'|$, which is $\hat{t} \geq |V'|/\polylog(n)$. By analyzing the recursion tree $T$, and that instances $V'$ in any single depth of $T$ are disjoint, we conclude that there cannot be more than $\tO(m)$ such queries in a single depth of $T$.

The second challenge is in the correctness proof for any node pair $u,v$. Similar to~\cite{AKT20}, we would like to reason about the pivot choice $p$ in the instance separating $u$ and $v$. However, vertex separators in our case introduce a worrisome dependency, as it could be for some instances that the only nodes separating $u$ and $v$ as pivots are $u$, $v$, and nodes within their separator $C_{u,v}$. In this case, a careful analysis is needed to prove that such node pairs are likely to be pushed together (i.e. not being separated) to a later stage. We do this by using a more global argument that is independent of the particular instance separating $u$ and $v$, described below.

On one hand, while $|V'|$ is large compared to $|C_{u,v}|$, there is a decent probability to pick pivots outside of $C_{u,v}$. 
On the other hand, if $|V'|$ becomes small enough while $u$ and $v$ are still together, then they would be queried as a pair of high degree nodes in $V'$. This motivates the following analysis (see Figure~\ref{Figs:prob}). Consider the path $P$ from the root of $T$ to $V'_u$, the highest recursion node in $T$ where $|V'_u|$ is smaller than a threshold set to $\tO(d_G(u))$, where we assume without loss of generality that $d(u)\leq d(v)$. We show that with constant probability, all pivot choices in instances along $P$ do not belong to $C_{u,v}$. Assuming that this event occurs, we have two cases. If $v$ has survived all the way down to $V'_u$, then $v$ and $u$'s high degrees compared to $|V'_u|$ would cause the algorithm to query them directly as a pair at this stage. Otherwise, they got separated earlier \emph{along $P$}, and consequently, the pivot chosen at the moment of separation does not belong to $C_{u,v}$.

\begin{figure}[!ht]
       \includegraphics[width=0.65\textwidth,center]{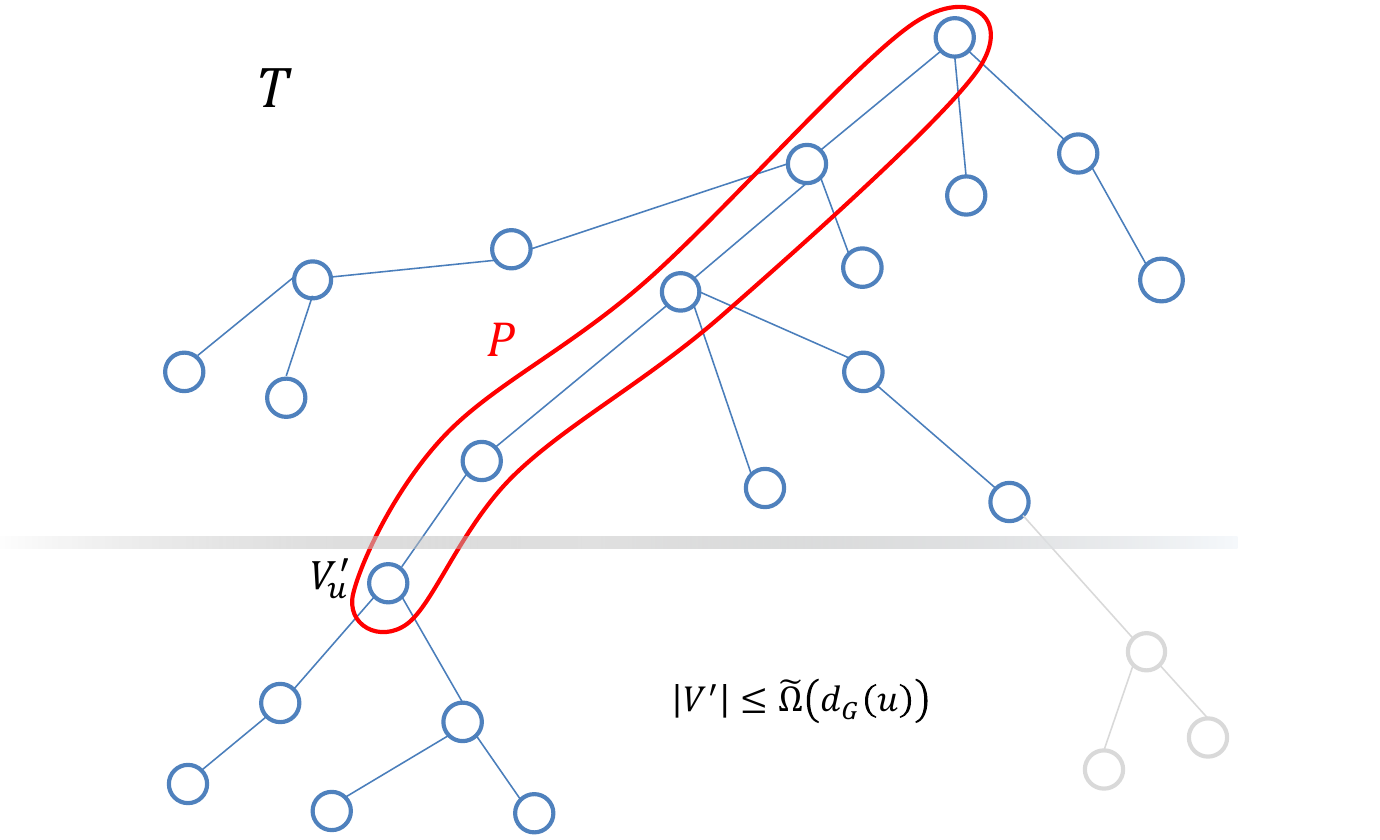}
   \caption[-]{
If $v\in V_{u}$ then both $u,v$ are queried directly as a high degree pair. Otherwise, with constant probability $v$ must have been separated from $u$ in a higher recursion node along $P$, in which case $v\notin C_{u,v}$ at that moment of separation.
   }
   \label{Figs:prob}
\end{figure}

Altogether, property (*) in~\cite{AKT20} is modified to fit our setting as follows:

\begin{itemize} \compactify
\item[(**)] \label{it:reinforced_starproperty}
  For every pair $s' \in S_{i,p}, t'\in S_{j,p}$ for $i \neq j$ such that $\MF(p,s')\leq \MF(p,t')$, with constant probability at least one of the following holds.
\begin{enumerate}
\item $t'\in C_{p,i}.$
\item $d_G(s'),d_G(t')\geq |V'|/(20\log_{4/3} n).$
\item\label{property:3} $(S_{p,i}, C_{p,i},S_{i,p})$ corresponds to a minimum $s't'$-cut in $G$.
\end{enumerate}
\end{itemize}

Finally, we run the algorithm $O(\log n)$ times, concluding that for every node pair, with high probability (**) occurs in at least one of the recursion trees.
\paragraph*{Converting to Las Vegas}
The algorithm above may, with non zero probability, produce erroneous answers. This occurs for a pair $u,v$ if the pivots separating $u$ and $v$ belong to $C_{u,v}$ in all $O(\log n)$ repetitions (and the pair $u,v$ was not queried directly).
A natural starting point to address this issue is to increase the number of repetitions to $x$ for some $x>>\log n$, instead of just $O(\log n)$ as in the Monte Carlo algorithm.
By doing so, we can argue that the \MF values that are at most $x$ will be preserved by a global analysis claim \emph{and verification} that no pivot was chosen too many times in instances of size at least $\tilde{\Omega}(x)$ (throughout all repetitions). As a result, a sufficient number of distinct pivots must have been chosen, ensuring that at least one of the pivots separating $u$ and $v$ does not belong to $C_{u,v}$.
This could make the basis of a correct Las Vegas algorithm; unfortunately, the running time becomes $xmT(m)$, which is too high as $x$ is only bounded by $n$. However, the part responsible for this high running time is the reinforcement queries rather than the basic recursion, and so there is still a hope.

Indeed, in the Las Vegas algorithm we are more strategic in performing reinforcement queries.
We would like to perform such queries only for pairs that were required in at least $x$ repetitions, thus bounding the total number of distinct queries throughout all repetitions by $\tO(m)$.
Accordingly, we would claim that if a pair $u,v$ was \emph{not} required in at least $x$ repetitions (and thus was not queried directly) then it is surely separated many times, and so by many distinct pivots, in the basic recursion structure.
Similar to before, this ensures that at least one of those pivots must not belong to $C_{u,v}$, making the corresponding answer correct.
However, even this only brings the total running time to $xnT(m)$ which equals $n^2m$ in the worst case.
To further speed up the algorithm, we operate in stages focusing on a range $[k,2k)$ of \MF values each time, thus incurring a factor $O(\log n)$ to the running time. 
Consequently, the relevant nodes, i.e,~\emph{terminals}, in each $[k,2k)$ stage are the nodes whose degrees are at least $k$, totaling at most $m/k$ nodes. At the same time, the minimum instance size we need to consider, as well as the number of repetitions required to ensure a sufficiently high number of distinct pivots for the terminals at this stage, is only $\tO(k)$. This brings the total running time to $\tO(m\cdot T(m))$, as required.

Note that unlike properties (*) and (**) detailed above, no local property can naturally be described in the Las Vegas version of our algorithm. However, a more global property is appropriate:

\begin{itemize} \compactify
\item[(***)] \label{it:reinforced_starproperty_}
  For every pair $s', t'$ such that $\MF(s',t')\in [k,2k)$, when the algorithm terminates at least one of the following holds in the $[k,2k)$ stage.
\begin{enumerate}
\item $s'\in S_{i,p}$ and $t'\in C_{p,i}$, or $s'\in C_{p,j}$ and $t'\in S_{j,p}$, in at least $k$ repetitions for some $i$'s and $j$'s. 
\item $s',t'$ are together in final instances of at least $k$ repetitions.
\item\label{property:3_} $(S_{p,i}, C_{p,i},S_{i,p})$ in at least one of the repetitions corresponds to a minimum $s't'$-cut in $G$.
\end{enumerate}
\end{itemize}

\subsubsection{Nondeterministic algorithms}\label{subsubsec:nondet}

Several papers~\cite{CGIMPS16,AKT20_soda,Li21} have proposed nondeterministic algorithms as a means to restrict the applicability of \SETH for conditional lower bounds. This method was introduced by Carmosino et al.~\cite{CGIMPS16}, who also demonstrated that several problems, including APSP and $3$SUM, are unlikely to have conditional lower bounds under \SETH, assuming that \NSETH holds.

\paragraph{A related nondeterministic algorithm in~\cite{AKT20_soda}.}
Abboud et al. \cite{AKT20_soda} presented a nondeterministic algorithm for the Gomory-Hu tree problem and used the framework of~\cite{CGIMPS16} to demonstrate that it is unlikely for this problem to have a super-linear \SETH-based conditional lower bound.
In their work, the nondeterministic algorithm simulates the Gomory-Hu algorithm~\cite{GH61} on a given input graph. Instead of using \MF routines like in the Gomory-Hu algorithm,~\cite{AKT20_soda} employs a new accelerated operation that is equivalent to determining the maximum $pv$-flows from a guessed pivot $p$ to multiple other nodes $v$ simultaneously.

To quickly verify these multiple maximum flows, they used succinct data structures that exist for their setting of undirected graphs with edge-capacities (i.e. they used a cut-equivalent tree to verify cuts and a tree packing to verify flows). By applying the min-cut max-flow theorem~\cite{FF56}, they concluded that if equality was established for each pair of guessed pivot $p$ and a corresponding node $v$, then the maximum flow guesses were correct. Consequently, the correctness of the simulation follows from the correctness of the Gomory-Hu algorithm.

\paragraph{Obstacles in accelerating current algorithms.}
In contrast, for hard variants even in unit (node or edge) capacities settings, existing algorithms~\cite{CLL13,HLSW23} have significant barriers to getting accelerated using nondeterminism.
In the $O(m^{\omega})$ time algorithm for directed graphs with unit edge-capacities~\cite{CLL13}, merely writing down (as opposed to finding) an inherent part of the algorithm -- an inverted $m$ by $m$ matrix -- takes $m^2$ time.
Another algorithm, proposed for undirected graphs with unit node-capacities~\cite{HLSW23}, utilizes an element-connectivity Gomory-Hu tree for nodes with smaller degrees and directly calls \MF for node pairs with high degrees. However, in the case of dense graphs, this algorithm performs similarly to the naive approach of computing \MF for every node pair. Moreover, even if element-connectivity Gomory-Hu trees can be found nondeterministically in linear time, the nondeterministic time complexity of \APMF in our setting would still be $O(m^2)$ due to a blow-up caused by set sampling techniques~\cite{IN12,PSY22,HLSW23} which seem necessary for element-connectivity based algorithms. As a result, alternative approaches are needed.

Let us focus on undirected graphs with unit node-capacities. Straightforward cut and flow witnesses for all pairs have a total size of $O(n^3)$
\footnote{Following previous results in similar contexts~\cite{CGIMPS16,AKT20_soda,Li21}, verifying both cuts and flows is crucial (1) to rule out \textit{Turing} reductions, and (2) as (minimum-) cuts are simply necessary for some nondeterministic algorithms to work (e.g., see~\cite{AKT20_soda}, and Theorem~\ref{thm:undirected_unit_node} here).}, 
and verifying the flows could be done easily within a similar time since a flow here is composed of internally node-disjoint paths. However, verifying the cuts seems more challenging as it requires reading all edges to ensure that for each guess cut, no edges exist between the two separated parts\footnote{For ease of exposition, we assume that there is no edge directly connecting the corresponding pair. However, this assumption is not necessary in any of our results.}, and it is not immediately apparent how to perform this task faster than the naive approach.

A natural attempt to verify the cuts could be to use connectivity oracles to test whether pairs $u,v$ are still connected after removing their guessed separator. However, this approach might not be promising either, as conditional lower bounds are known for similar problems, such as the $d$-vertex-failure connectivity and even the vertex-failure global connectivity oracle~\cite{LS22}.

\paragraph{Our approach.}
We enter all of the guessed cuts into a large binary matrix of size $n^2\times n$, and then multiply it by the incidence matrix of the input graph. The resulting matrix enables us to confirm that there are no edges between the two parts in any of the guessed cuts, indicating that they are all valid vertex-cuts. 
This process takes the time of rectangular matrix multiplication, that is $n^{\omega(2,1,1)}\leq O(n^{3.252})$~\cite{GU18}.
In fact, we consider a slightly more general version of this approach, in two ways. 
First, it can be used with any algorithm for \APMF, not just the naive approach. This is beneficial for algorithms that require fewer than $n^2$ queries to \MF on the input graph.
Second, our approach can work with any data structure for \APMF, rather than just the trivial one that outputs all $n^2$ values. This expands the set of potential algorithms to include those outputting data structures whose size is smaller than the output size of \APMF, which is $n^2$.
Moreover, this approach is robust enough to work even on directed graphs with general edge-capacities. Therefore, for subquartic algorithms it essentially settles the nondeterministic time complexity of verifying the cuts in all variants of \APMF via folklore reductions with little overhead.

Regarding flows, the $O(n)$ bound naturally applies only for unit node-capacities. Are there nontrivial upper bounds in other settings?
For general node-capacities, we answer this question affirmatively! 
Specifically, we show that for graphs with general node-capacities, even directed, $O(n)$ edges are always sufficient to ship the maximum $st$-flow. This is done by identifying a certain pattern that can be eliminated from any flow while still preserving its value and without violating the flow constraints or increasing the volume
\footnote{An alternative proof uses linear programming arguments~\cite{Che24}, as discussed later in this paper.}.
For directed (multi-) graphs with unit \textit{edge}-capacities,
we use the upper bound of $n\sqrt{\val(f_{st})}$ proved in~\cite{KL98} for the volume of $st$-flows $f_{st}$, where $\val(f_{st})$ represents the amount of flow shipped from $s$ to $t$ in $f_{st}$. 
Using these flow bounds, a nondeterministic algorithm for each variant can guess succinct maximum flows for the queried pairs whose support is bounded by the corresponding upper bounds, and then verify their feasibility and values.
Finally, we believe that our results may also be relevant in other related fields.

\subsection{Preliminaries}\label{prems}

\paragraph*{General notations.}

\begin{itemize}
\item
(\emph{Miscellaneous.})
Throughout, we only consider connected graphs with $m \geq n-1$, and we use the terms “node” and “vertex” interchangeably. 
Let $G=(V,E,c)$ be a graph with capacity function $c$ defined on either its edges or nodes. For positive integers $a$,$b$, and $c$, we denote by $\omega(a,b,c)$ the exponent such that $n^{\omega(a,b,c)}$ is the time it takes to multiply an $n^a\times n^b$ matrix by an $n^b\times n^c$ matrix.
\item
(\emph{Flows.}) For an $st$-flow $f_{st}$ in $G$, we use $\vol(f_{st})$ to denote the size of the support of $f_{st}$ in $G$, which is simply the number of edges in $G$ with nonzero flow shipped by $f_{st}$. We use $\val(f_{st})$ to denote the total flow value shipped from $s$ to $t$ in $f_{st}$, that is the (capacity of) the out-degree minus the in-degree of $s$ (equivalently, the in-degree minus the out-degree of $t$) in $f_{st}$.
We may assume without loss of generality that $f_{st}$ is acyclic (in particular, the in-degree of $s$ and the out-degree of $t$ is $0$), as removing cycles in $f_{st}$ does not change $\val(f_{st})$, which we typically attempt to maximize, and can only decrease $\vol(f_{st})$, which we typically attempt to minimize.
We denote by $V(f_{st})$ the set of nodes in the support of  $f_{st}$, i.e., nodes with adjacent edges carrying non-zero flow.
Depending on the context, we may use $f_{st}$ to denote the support of the corresponding flow.
\item
(\emph{Cuts.})
Given two node sets $A,B\subseteq V$ (which could also be a singletons), we use $E(A,B)$ to denote the set of edges between $A$ and $B$, and $\val(A,B)$ to denote the capacity of edges between $A$ and $B$, i.e., $\val(A,B)=\val(E(A,B))=\sum_{a\in A, b\in B} c(a,b)$.
An edge $uv$-cut is denoted by $(S_{u,v},S_{v,u})$, and a vertex $uv$-cut by $(S_{u,v},C_{u,v},S_{v,u})$, where $C_{u,v}$ is called a (vertex) $uv$-separator. 
In these contexts, $\val(C_{u,v})$ is defined as the capacity of the separator, given by $\sum_{w\in C_{u,v}} c(w)$.
\item
(\emph{Max-Flow.})
We use the notation $\MF_G(s,t)$ for the minimum $st$-cut, or \MC for short, in $G$. Depending on the context, this could either be the actual value or a query to an algorithm for computing it. In some cases, we may omit the graph $G$ subscript or the input $(s,t)$ if it is clear from the context.
\end{itemize}

\paragraph*{Min-Cut data structure.}
Following notations from~\cite{AKT20}, a \emph{Min-Cut data structure} for a graph family $\mathcal{F}$ is a data structure that given a graph $G\in \mathcal{F}$, after a preprocessing phase which makes $t_{pq}(m)$ \MF calls on $G$ and spends $t_{po}(m)$ time outside of these calls,
can answer Min-Cut queries for any two nodes $s,t\in V$ in amortized query time (or output sensitive time) $t_{mc}(k_{st})$, where $k_{st}$ denotes the output size (number
of edges or nodes in the output cut or separator). 
In particular, it means that the preprocessing time is at most $t_{pq}(m)\cdot \hO(m)+t_{po}(m)$ using~\cite{CKLP22,BCG+23}, and the query time is 
$O(k_{st}\cdot t_{mc}(k_{st}))$.
We may also be interested in Min-Cut data structures that support queries of \MF values instead of cuts, and we denote their query time by $t_{mf}(m)$.

\paragraph*{Non-reducibility.}
Using notations from~\cite{AKT20_soda}, a search problem $P$ is a binary relation,
and we say that $S$ is a solution to instance $x$ iff $(x,S)\in P$. 
Let $\SOL(x)=\set{S: (x,S)\in P}$ denote the set of solutions for instance $x$. 
We say that $P$ is a \emph{total function} if every instance $x$ has at least one solution, i.e. $\SOL(x)\neq \emptyset$.
Let $\bot$ be the ``don't know'' symbol and assume $\bot \notin \SOL(x)$ for all $x$.
In particular, in our context of \MCDSs, $x$ is a graph, and $\SOL(x)$ is the set of all \MCDSs with $t_{mc}(m)=\hO(1)$ for $x$ where $m$ is the number of edges in $x$.
Next, we define the nondeterministic complexity of a total function.
\begin{definition}[Definition $3.1$ from~\cite{AKT20_soda}. See also~\cite{CGIMPS16}]
We say that a total function $P$ has nondeterministic time complexity $T(n)$ if there is a deterministic Turing Machine $M$ such that for every instance $x$ of $P$ with size $\card{x}=n$:
\begin{enumerate}[label=\alph*.]
\item \label{it:time}
  $\forall g, \DTIME(M(x,g))\leq T(n)$,
  i.e., the time complexity of $M$ is bounded by $T(n)$; 
\item \label{it:a}
  $\exists g, M(x,g)\in \SOL(x)$, i.e., at least one guess leads $M$ to output a solution;
\item \label{it:b}
  $\forall g, M(x,g) \in \{\bot\}\cup \SOL(x)$, i.e., every guess leads $M$ to output either a solution or ``don't know''.
\end{enumerate}
\end{definition}

The following is a definition of (deterministic) fine-grained reductions from a decision problem to a total function. These are Turing reductions.
\begin{definition}[Definition $3.12$ from~\cite{AKT20_soda}. See also~\cite{WILL18}] 
Let $L$ be a language and $P$ be a total function, and let $T_L(\cdot)$ and $T_P(\cdot)$ be time bounds. We say that $(L,T_L)$ admits a fine-grained reduction to $(P,T_P)$ if for all $\varepsilon>0$ there is a $\gamma>0$ and a deterministic
Turing machine $M^P$ (with an access to an oracle that generates a solution to every instance of $P$) such that:
\begin{enumerate}
\item $M^P$ decides $L$ correctly on all inputs when given a correct oracle for $P$.
\item Let $\tilde{Q}(M^P,x)$ denote the set of oracle queries made by $M^P$ on input $x$ of length $n$. Then the query lengths obey the bound 
$$
\forall x, \qquad
\DTIME(M^P,\card{x}) + \sum_{q\in \tilde{Q}(M,x)}(T_P(\card{q}))^{1-\varepsilon}
\leq(T_L(n))^{1-\gamma}.
$$
\end{enumerate}
\end{definition}
Theorem $3.13$ in~\cite{AKT20_soda}, provided below, further showed a non-reducibility result under \NSETH for total functions with small nondeterministic complexity.
\begin{theorem}\label{thm:nseth_total} 
Suppose $P$ is a total function with nondeterministic time complexity $T(m)$. If for some $\delta>0$, there is a deterministic fine-grained reduction from \kSAT with time-bound $2^n$ to $P$ with time bound $T(m)^{1+\delta}$, i.e. from $(\text{\kSAT},2^n)$ to $(P,T(m)^{1+\delta})$ then \NSETH is false.
\end{theorem}

\paragraph*{Non-reducibility for deterministic vs. randomized reductions.}
Our results, along with the framework for non-reducibility in~\cite{CGIMPS16}, do not address the possibility of proving lower bounds based on \SETH using a randomized fine-grained reduction. The reason is that \NSETH, upon which this framework is built, does not remain plausible when facing randomization (see~\cite{CGIMPS16,Wil16_MASETH}). That said, we are not aware of any instances where this barrier has been successfully bypassed with randomization.
Furthermore, one can use a non-uniform variant of \NSETH called \NUNSETH, under which our results also hold for randomized reductions with bounded error.

\section{Non-Reducibility Results}

In the absence of $n^{4-o(1)}$ conditional lower bounds for general algorithms for \APMF, it may be tempting to utilize the commonly used \SETH as a hardness assumption, and design a reduction from SAT to prove that the problem requires a running time of $n^{4-o(1)}$. However, we demonstrate here that in most settings of the problem, such (deterministic) reductions are not possible under \NSETH, thus proving our Theorems~\ref{thm:undirected_unit_node}, \ref{thm:nodecap}, and \ref{thm:unit_edge}.
Technically, our main results in this section are fast \emph{nondeterministic} algorithms for solving \APMF in various settings.
In fact, our nondeterministic framework is even more general, as it is applicable to \MCDSs. 
This could be significant because in some settings, a data structure for \APMF could potentially be much smaller than $n^2$ -- similar to cut-equivalent trees in the case of undirected graphs with edge capacities, and we aim to capture these scenarios as well with our framework.

Then, we reach our non-reducibility conclusions under \NSETH by an argument first made in~\cite{CGIMPS16} to rule out \SETH-based reductions to various problems, including $3$SUM and All-Pairs Shortest-Path.
The argument was later rewritten~\cite{AKT20_soda} to adapt from decision problems or functions (where each input has exactly one output) to total functions
as every graph has at least one Gomory-Hu tree, which was the focus of~\cite{AKT20_soda}. This adaptation is detailed in Section~\ref{prems}, with Theorem~\ref{thm:nseth_total}.
We use the same total function definition more generally for any \MCDS with $t_{mc}=\hO(1)$, which is a total function since every graph has at least one such data structure: just compute the $O(n^2)$ cuts naively in $O(n^2m)$ time (so in this case, $t_{pq}(m),t_{po}(m)=n^2$).
Next is our main framework, which is central to achieving our nondeterministic algorithms across all settings.
\begin{theorem}\label{thm:nondet}
Suppose there is a \MCDS $\mathcal{D}$ that given a graph $G\in\mathcal{F}$ on $m$ edges and $n$ nodes, after a preprocessing phase which, with high probability, makes $t_{pq}(m)$ \MF calls on $G$ and spends $t_{po}(m)$ time outside of these calls, can answer queries correctly in $t_{mc}(m)$ output sensitive time.

Then, the nondeterministic time complexity for constructing a \MCDS for $G$ is $\TPM=O(t_{po}(m)+n^{\omega(\log_{n}(|Q|),1,1)}+\sum_{(u,v)\in Q}\vol(f_{uv}))$ with $\TMCM=t_{mc}(m)$,
where $Q\subseteq V\times V$ is the set of pairs for which the $t_{pq}(m)$ \MF calls are made, and $f_{uv}$ is any maximum $uv$-flow in $G$.
\end{theorem}

The main idea in the proof of Theorem~\ref{thm:nondet} is to first verify all the minimum cuts that $\mathcal{D}$ will use on $G$ in hindsight. Then, we simulate $\mathcal{D}$ and whenever we encounter a \MF query, we assert the corresponding guess. The correctness follows immediately from the correctness of $\mathcal{D}$. Before proving Theorem~\ref{thm:nondet}, we state a technical lemma regarding the nondeterministic time it takes to verify maximum flows for a given set of node pairs $Q$.

\begin{lemma}\label{lemma:nondetmaxflows}
Given a set of node pairs $Q\subseteq V\times V$, 
the nondeterministic complexity of solving \MF for all $(u,v)\in Q$ is 
$O\left(n^{\omega(\log_{n}(|Q|),1,1)}+\sum_{(u,v)\in Q}\vol(f_{uv})\right)$.
\end{lemma}

\begin{proof}[Proof of Theorem~\ref{thm:nondet}]
A nondeterministic algorithm first guesses the 
random bits of $\mathcal{D}$ (if $\mathcal{D}$ is randomized Las Vegas) and a
set of node pairs to be queried $Q$.
It then solves $Q$ nondeterministically using Lemma~\ref{lemma:nondetmaxflows}.
Finally, it simulates $\mathcal{D}$ (on a deterministic Turing machine) with the guessed random bits, and whenever encountering a \MF query, it asserts the validity of the corresponding stored answer for the node pair.
By Lemma~\ref{lemma:nondetmaxflows}, verifying the \MF for $Q$ takes a total time of 
$O\left(n^{\omega(\log_{n}(|Q|),1,1)}+\sum_{(u,v)\in Q}\vol(f_{uv})\right)$, and otherwise (for some guess) the total running time is bounded by $t_{po}(m)$.
The correctness follows immediately from the correctness of $\mathcal{D}$, and so the proof of Theorem~\ref{thm:nondet} is concluded.
\end{proof}

\begin{proof}[Proof of Lemma~\ref{lemma:nondetmaxflows}]
First, guess a list $\{p_{uv}\}_{(u,v)\in Q}$ for the node pairs in $Q$.
We want to verify quickly two things for each pair $(u,v)$:
\begin{enumerate}
\item\label{flows} that there exists a $uv$-flow $f_{uv}$ such that $\val(f_{uv})\geq p_{uv}$, and
\item\label{cuts} that there exists a $uv$-cut $(S_{uv},T_{uv})$ such that $\val(S_{uv},T_{uv})\leq p_{uv}$.
\end{enumerate}
By the min-cut max-flow theorem~\cite{FF56}, it would imply that the guessed $p_{uv}$'s are maximum $uv$-flow values.
But how can we verify these conditions quickly? Verifying the cut or flow values explicitly for each pair would be too costly, as it might take at least $m$ time using a straightforward approach. Thus, the total time would be $O(|Q|\cdot m)$, which does not provide any significant advantage to nondeterminism. Therefore, we need to employ a more careful strategy.

\paragraph{Verifying lower bounds.}
For the flow lower bounds, verifying the guessed values $\{p_{uv}\}_{(u,v)\in Q}$ in the proper time of $\sum_{(u,v)\in Q} \vol(f_{uv})$ is simple; just guess succinct flows $f_{uv}$ for all pairs $(u,v)\in Q$, and then verify their validity and values $p_{uv}$, i.e., that the flow conservation and capacity constraints hold, and that $\val(f_{uv})\geq p_{uv}$. If any verification step returned "no", reject, otherwise accept. 
The verification time per flow is proportional to the number of edges in its support, that is $O(\vol(f_{uv}))$, as required.

As a side note, the challenge in making this part of Lemma~\ref{lemma:nondetmaxflows} useful 
will be to show for any maximum $uv$-flow $f_{uv}$ that \emph{there exists} a (perhaps alternative) succinct flow $f'_{uv}$ shipping the same amount of flow in $G$ while keeping the flow constraints, or to find a proper algorithm with few calls to \MF.
This will be elaborated on in Section~\ref{subsec:succinct_flows}.

\paragraph{Verifying upper bounds.}
For the cut upper bounds, the proof is similar across all variants, and it can even be applied to directed graphs with general \textit{edge} capacities.
Therefore, we will only present a proof for this particular variant, while noting that folklore reductions from the other variants are sufficient for our purposes.

Technically, the nondeterministic algorithm guesses a partition $(S_{uv},T_{uv}=V\setminus S)$ for every pair $(u,v)\in Q$, and then it utilizes fast matrix multiplication to simultaneously verify that the cut values are smaller than the guessed values $\{p_{uv}\}_{(u,v)\in Q}$, i.e. it checks that for all $(u,v)\in Q$, it holds that $\val(S_{uv},T_{uv})\leq p_{uv}$.
In more detail, we define a binary matrix $P_{|Q|\times n}$, where the rows correspond to the pairs in $Q$ and the columns correspond to the nodes in the graph. In this matrix, the entry $P[g_{uv},w]$ is set to $1$ if $w\in S_{uv}$, and to $0$ otherwise, where $g_{uv}$ is the index of the pair $(u,v)\in Q$ (by some arbitrary order).
Lastly, let $M_G$ be the incidence matrix of the graph $G$, where an entry $M_G[u,v]$ is the weight of the edge from $u$ to $v$ in $G$, and is $0$ if no such edge exists. With these notations, we are now ready to describe the cuts verification algorithm.

\paragraph{Algorithm: Cuts Verification}
\begin{enumerate}
\item perform the multiplication $P'=PM_G$.
\item for every pair $(u,v)\in Q$ verify that 
$\sum_{z\in T_{uv}} P'[g_{uv},z]\leq p_{uv}.$
\item if at least one verification step returned "no", reject, otherwise accept.
\end{enumerate}

\paragraph{Correctness.}
Observe that for every node $z$, $P'[g_{uv},z]$ is precisely 
$\sum_{w\in S_{uv}} c(w,z)$, and thus for each pair $(u,v)\in Q$, the verified values are precisely what we are looking for:  
$$
\sum_{z\in T_{uv}} P'[g_{uv},z]=
\sum_{w\in S_{uv},z\in T_{uv}} c(w,z)= \val(S_{uv},T_{uv}).
$$

\paragraph{Running time.}
The total size of $|Q|$ cut guesses (by their nodes) is $O(|Q|\cdot n)$. Multiplying an $|Q|\times n$ matrix with an $n\times n$ matrix takes $n^{\omega(\log_{n}(|Q|),1,1)}$ time, and counting $O(n)$ entries per pair of nodes in $Q$ takes $O(|Q|\cdot n)$ time, which brings the total verification time
to $O\left(n^{\omega(\log_{n}(|Q|),1,1)}+|Q|\cdot n\right)\leq O\left(n^{\omega(\log_{n}(|Q|),1,1)}\right),$
concluding Lemma~\ref{lemma:nondetmaxflows}.
\end{proof}
\subsection{Succinct $st$-flows in Various Settings}\label{subsec:succinct_flows}
In this section, we detail the volume upper bounds for flows in various settings, that we either prove or reference.

\subsubsection{Unit node-capacities}

\begin{lemma}\label{lemma:flows:unit_node}
Let $G=(V,E)$ be an undirected graph with unit node capacities, $s,t\in V$ a pair of nodes, and $f_{st}$ an $st$-flow in $G$. Then $\vol(f_{st})\leq 2n$.
\end{lemma}
\begin{proof}
Observe that $f_{st}$ consists of internally node-disjoint paths. Therefore, each node in $V(f_{st})$ except for $s$ and $t$, has degree $2$ in $f_{st}$, and consequently we have $\vol(f_{st})\leq 2n$, as required.
\end{proof}

\subsubsection{Unit edge-Capacities}
In contrast to the case of unit node-capacities, here, nodes in $V(f_{st})$ can have arbitrarily high degrees in $f_{st}$, and it is not immediately evident that the total volume of $f_{st}$ is bounded by anything significantly smaller than $m$.
We use a general lemma from~\cite{KL98} that works even for arbitrary integer capacities, in our context of (multi-)graphs with unit edge-capacities:

\begin{lemma}[Theorem $7.1$ in~\cite{KL98}, rephrased]\label{lemma:flows:unit_edge}
Let $G=(V,E)$ be a directed graph with integer edge-capacities, $s,t\in V$ a pair of nodes, and $f_{st}$ an $st$-flow in $G$. Then $\vol(f_{st})\leq \tO\left(n\sqrt{\val(f_{st})}\right)$. 
\end{lemma}

\subsubsection{General node-capacities}
We will now prove that the total volume necessary to ship any flow value in a node-capacitated network never exceeds $2n$.

\begin{lemma}\label{lemma:flows:nodecap}
Let $G=(V,E)$ be a directed graph with node-capacities, $s,t\in V$ be a pair of nodes, and $f_{st}$ be an $st$-flow in $G$. Then there exists an $st$-flow $f'_{st}$ in $G$ of the same value $\val(f'_{st})=\val(f_{st})$ such that $\vol(f'_{st})\leq 2n$.
\end{lemma}

The high-level idea of the proof is that any valid $st$-flow $f_{st}$ can be transformed into a succinct $st$-flow $f'_{st}$ by eliminating all copies of a certain pattern from $f_{st}$, while preserving the flow value and meeting the flow constraints.
The pattern that will be eliminated is called an \emph{anti-directed closed walk} (see Figure~\ref{Figs:anti_directed_cycle}). It is, in our context, a directed graph whose underlying graph forms a \textit{simple} closed walk, that is where no edge appears twice, and consecutive edges have alternate directions, or in other words, no pair of consecutive edges forms a directed path.
Note that an anti-directed closed walk must have an even number of edges.
\begin{definition}[anti-directed closed walk]
For any integer $\ell$, an \emph{anti-directed closed walk} is a simple closed walk $C=(u_1,u_2,\ldots,u_{2\ell})$ where for each $i\in [1,\ldots,2\ell]$, the edges between $u_{i-1},u_i$ and between $u_i,u_{i+1}$ are oriented in one of two ways: either (1) towards $u_i$, that is $(u_{i-1},u_i), (u_{i+1},u_i)\in E(C)$, or (2) away from $u_i$, that is $(u_i,u_{i-1}), (u_i,u_{i+1})\in E(C)$.
\end{definition}

For an anti-directed closed walk $C$ we define $P_{out}$ (respectively, $P_{in}$) as the set of all node occurrences $u_i\in C$ whose adjacent edges are oriented in way (1) (respectively, in way (2)). 
Note that $|P_{in}|=|P_{out}|$, $P_{in}\cap P_{out}=\emptyset$, and $P_{in}\cup P_{out}=C$. 

The related concept of anti-directed \textit{cycles} has been studied before in other contexts, such as edge partitioning~\cite{AHR81} and the existence of anti-directed hamiltonian cycles~\cite{DM15}.
We are now ready to prove Lemma~\ref{lemma:flows:nodecap}.
\begin{proof}[Proof of Lemma~\ref{lemma:flows:nodecap}]
Let $f_{st}$ be an $st$-flow in $G$. To modify $f_{st}$, we iteratively find an anti-directed closed walk $C$ in $f_{st}$ and reroute the flow of $f_{st}$ that is shipped through $C$. 
The rerouting is done in such a way that results in a zero flow through at least one edge of $C$, causing the edge to be removed from $f_{st}$. This is achieved while ensuring that the total flow entering each node in $P_{in}$ and leaving each node in $P_{out}$ remains exactly the same before and after the operation (refer to Figure~\ref{Figs:anti_directed_cycle}).

\begin{figure}[!ht]
       \includegraphics[width=1.0\textwidth,left]{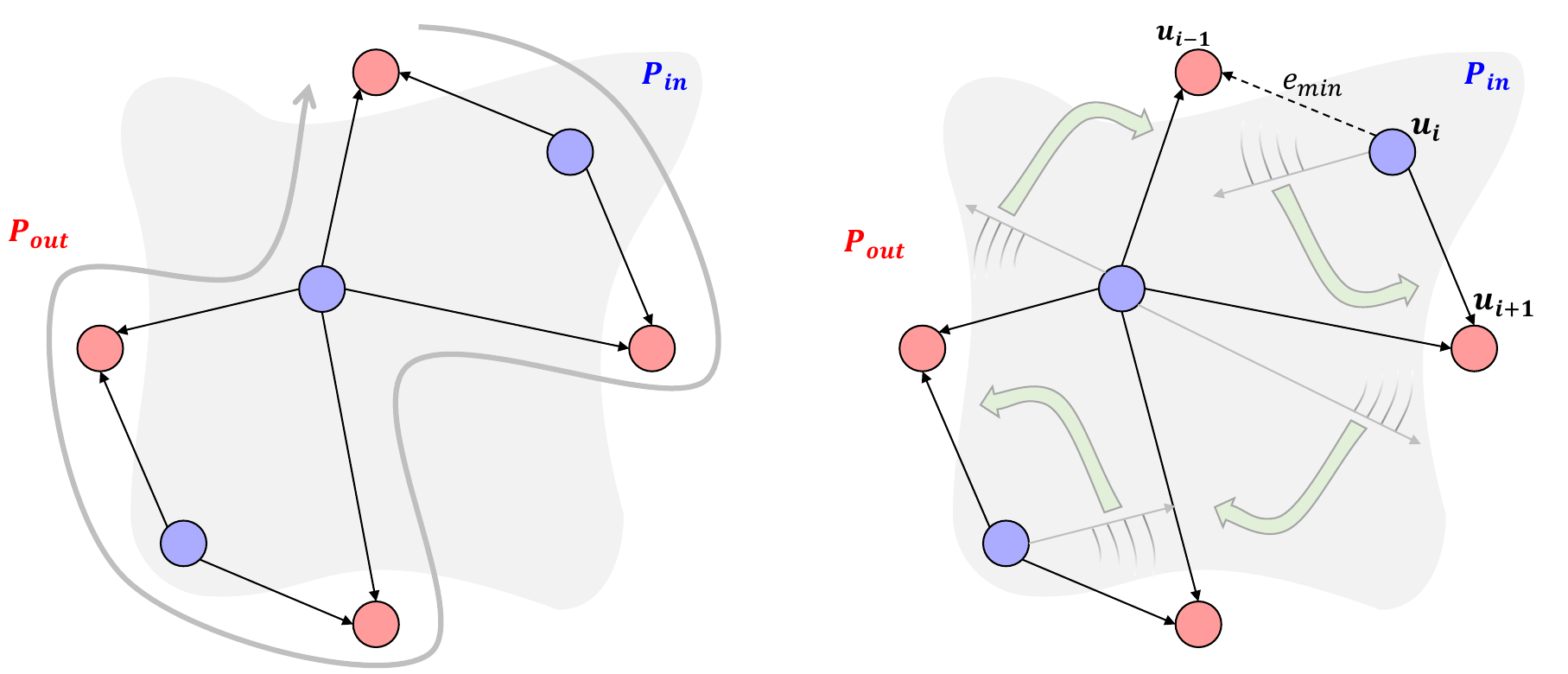}
   \caption[-]{
The figure on the left shows an anti-directed walk of length $8$, while the figure on the right illustrates its elimination.
   }
   \label{Figs:anti_directed_cycle}
\end{figure}

More formally, let $e_{min}$ be an edge in $C$ such that the flow of $f_{st}$ through it $f_{st}(e_{min})$ is minimal among all edges of $C$, and assume without loss of generality that $e_{min}$ is oriented from a node $u_i$ to its predecessor $u_{i-1}$ in $C$ for some $i\in [|C|]$ (where indices are calculated modulo $|C|$).
For every node occurrence $u_j\in P_{in}$, we reroute a flow of value $f_{st}(e_{min})$ that was previously shipped through the edge $(u_j,u_{j-1})$ to flow through the edge $(u_j,u_{j+1})$ instead.
By induction on the number of rerouting operations, $f_{st}$ remains valid after this operation, meaning that it (1) respects the edge directions, and (2) adheres to the node capacities and flow conservation constraints.
This is because while $u_{j+1}\in P_{out}$ receives additional flow of $f_{st}(e_{min})$ from its predecessor $u_{j}\in P_{in}$, it also receives less flow by the exact same amount from its successor $u_{j+2}\in P_{in}$.
Furthermore, after the operation, the flow shipped through $e_{min}$ in $f_{st}$ becomes zero, and therefore $e_{min}$ does not exist in $f_{st}$ anymore, and $C$ is no longer an anti-directed closed walk there.

Finally, we demonstrate that when the operation can no longer be performed in $f_{st}$, i.e., no more anti-directed closed walks exist in $f_{st}$, then the resulting $st$-flow $f'_{st}$ satisfies $\vol(f'_{st})\leq 2n$. Since each operation reduces $\vol(f_{st})$ by at least $1$, this must occur at some point.
To begin with, we consider the standard reduction of $f'_{st}$ to a bipartite graph $G_{1,2}=(V_1\cup V_2, E_{1,2})$. $V_1$ and $V_2$ are two copies of $V(f'_{st})$, and there is an edge in $E_{1,2}$ from $u_1\in V_1$ to $v_2\in V_2$ iff the edge $(u,v)$ exists in $f'_{st}$. Note that any cycle in the undirected graph underlying $G_{1,2}$ corresponds to an anti-directed closed walk in $f'_{st}$.
Since no more anti-directed closed walks exist in $f'_{st}$, we conclude that the undirected graph underlying $G_{1,2}$ is a forest, and therefore, $\vol(f'_{st})\leq |E_{1,2}|\leq 2|V(f'_{st})|\leq 2n$, as required by Lemma~\ref{lemma:flows:nodecap}.
\end{proof}

Below, we sketch an alternative proof that uses linear programming~\cite{Che24}. First, we write \MF as a linear program, as usual, but adjusted to vertex capacities. 

\begin{equation*}
\begin{array}{ll@{}ll}
\text{maximize}  & \displaystyle\sum_{e \in \delta^+(s)} f(e) - \sum_{e \in \delta^-(s)} f(e) &\\
\text{subject to}
&\displaystyle\sum_{e \in \delta^-(v)} f(e)=\sum_{e \in \delta^+(v)} f(e), &\hspace{1cm}\forall v \in V\setminus\{s,t\}\\
&\displaystyle\sum_{e \in \delta^-(v)} f(e) \le c(v), &\hspace{1cm}\forall v \in V\setminus \{s, t\}\\
& f(e)\geq 0, &\hspace{1cm}\forall e\in E
\end{array}
\end{equation*}

In the above, $\delta^-(v)$ is the set of incoming edges to $v$, and $\delta^+(v)$ is the set of outgoing edges from $v$.
Note that there are $m$ variables, $n-2$ flow conservation constraints and $n-2$ node-capacity constraints (since $s$ and $t$ do not have capacity constraints), and $m$ non-negativity constraints.
From standard polyhedral theory (see the proof of Proposition $2.4$ in~\cite{CK04}, as well as~\cite{Raj90}), the linear program above admits an optimal solution with as many tight-constraints as variables, namely $m$. 
Since the total number of nontrivial constraints is $2n-4$, there exists an optimal solution with at most this many non-zero variables.

\subsection{Nondeterministic Algorithms in Various Settings}
Finally, we will use Theorem~\ref{thm:nondet} to show nondeterministic algorithms with time tailored to each setting; Theorems \ref{thm:undirected_unit_node}, \ref{thm:nodecap}, \ref{thm:unit_edge} will immediately follow from Lemmas \ref{lemma:nondet:undirected_unit_node}, \ref{lemma:nondet:nodecap}, \ref{lemma:nondet:unit_edge}, respectively, by applying Theorem~\ref{thm:nseth_total}.

\begin{lemma}\label{lemma:nondet:undirected_unit_node}\emph{(Undirected graphs with unit node-capacities.)}
The nondeterministic time complexity of \MCDS on undirected graphs with unit node-capacities is 
$\TPM=O\left(n^{\omega(\log_n m,1,1)}\right)$ with $\TMC=O(1)$.
\end{lemma}
Assuming $\omega=2$, the time becomes $\TPM=\tO(mn)$.
\begin{proof}[Proof of Lemma~\ref{lemma:nondet:undirected_unit_node}]
By combining Theorem~\ref{thm:nondet} with the flow bound in Lemma~\ref{lemma:flows:unit_node} and our Las Vegas algorithm in Theorem~\ref{thm:algorithm:main} which has $t_{pq}(m)=\tO(m),t_{po}=\tO(mn)$, and $t_{mc}=\tO(1)$, we get a nondeterminstic time of $\TPM = O\left(n^{\omega\left(\log_n m,1,1\right)}\right)$ with $\TMC=\tO(1)$, as required by lemma~\ref{lemma:nondet:undirected_unit_node}.
\end{proof}

\begin{lemma}\label{lemma:nondet:nodecap}\emph{(Directed graphs with general node-capacities.)}
The nondeterministic time complexity of \MCDS on directed graphs with general node-capacities is $\TPM=O\left(n^{\omega(2,1,1)}\right)$ with $\TMC=O(1)$.
\end{lemma}
Assuming $\omega=2$, the time becomes $\TPM=O(n^3)$.
With current algorithms, the time is $\TPM=O(n^{3.252})$ using rectangular matrix multiplication~\cite{GU18}. 
\begin{proof}[Proof of Lemma~\ref{lemma:nondet:nodecap}]
By combining Theorem~\ref{thm:nondet} with the flow bound in Lemma~\ref{lemma:flows:nodecap} and the naive algorithm which has $t_{pq}(m),t_{po}(m)=O(n^2)$ and $t_{mc}(m)=O(1)$, we get a nondeterminstic time of $\TPM=O\left(n^{\omega(2,1,1)}\right)$ with $\TMC=\tO(1)$, as required by Lemma~\ref{lemma:nondet:nodecap}.
\end{proof}
\begin{lemma}\label{lemma:nondet:unit_edge}\emph{(Directed graphs with unit edge-capacities.)}
The nondeterministic time complexity of \MCDS on directed graphs with unit edge-capacities is $\TPM=O\left(n^{5/2}\sqrt{m} + n^{\omega(2,1,1)}\right)$ with $\TMC=O(1)$.
\end{lemma}
Assuming $\omega=2$, the time becomes $\TPM=O\left(n^{5/2}\cdot \sqrt{m}\right)$.
With current algorithms, the time is $\TPM=O\left(n^{5/2}\cdot \sqrt{m} + n^{3.252}\right)$ using an algorithm for rectangular matrix multiplication~\cite{GU18}.
\begin{proof}[Proof of Lemma~\ref{lemma:nondet:unit_edge}]
According to Lemma~\ref{lemma:flows:unit_edge}, the flow part in the nondeterminstic time $\TPM$ in Theorem~\ref{thm:nondet} is
\begin{align*}
\sum_{(u,v)\in Q} \vol(f_{uv}) \leq 
\sum_{(u,v)\in Q} n\sqrt{\val(f_{uv})} \leq 
\sum_{(u,v)\in Q} n\cdot \sqrt{\min\{d_{out}(u),d_{in}(v)\}}.
\end{align*} 
By combining Theorem~\ref{thm:nondet} with the flow bound in Lemma~\ref{lemma:flows:unit_edge} and the naive algorithm which has $t_{pq}(m),t_{po}(m)=O(n^2)$, we get a nondeterministic time bound of
\begin{align*}
\TPM= O\left(n^2\cdot n\cdot \sqrt{m/n} + n^{\omega(2,1,1)}\right)=
O\left(n^{5/2} \sqrt{m} + n^{\omega(2,1,1)}\right),
\end{align*}
with $\TMC=\tO(1)$, as required by Lemma~\ref{lemma:nondet:unit_edge}.
\end{proof}

\section{All-Pairs Max-Flow for Unit Node-Capacities}

In this section, we prove a generalized version of Theorem~\ref{thm:algorithm:main} that accounts for node-capacities. We use parameter $\alpha:V\rightarrow [n]$ to represent for a node $v\in V$ an upper bound on the size of the minimum $uv$-separator $|C_{u,v}|$, for all nodes $u\in V$.\footnote{Via folklore reductions it captures, in particular, minimum vertex-cuts that in addition have smallest separator sizes.}
$\alpha$ is defined recursively, such that for any subset of nodes $A\subseteq V$, $\alpha(A)$ is equal to the sum of $\alpha(u)$ over all nodes $u\in A$. It is important to note that in the case of unit node-capacities, we have $\alpha(u)\leq d(u)$. This implies that $\alpha(V)\leq 2m$, which concludes Theorem~\ref{thm:algorithm:main}.

\begin{theorem}\label{thm:algorithm:main:gen}
There is a randomized Las Vegas algorithm that given an undirected graph $G=(V,E)$ with unit node-capacities, makes $t_{pq}(m)=\tO(\alpha(V))$ calls to \MF in $G$ and spends $t_{po}(m)=\tO(n\alpha(V))$ time outside of these calls, and can:
\begin{enumerate}
\item\label{DS:with} construct a data structure of size 
$\tO\left(\alpha(V)^{3/2}\right)$ that stores a set $\mathcal{C}$ of $\tO(\alpha(V))$ vertex cut-separators, such that given a queried pair $u,v\in V$ outputs in $t_{mc}=\tO(1)$ time a pointer to an optimal $uv$-separator in $\mathcal{C}$.
\item\label{DS:without} construct a data structure of size $\tO(\alpha(V))$ that given a queried pair $u,v\in V$ outputs in $t_{mf}=\tO(\MF(u,v))$ time the value $\MF(u,v)$ (or simply in $t_{mf}=\tO(1)$ time in an alternative Monte Carlo version).
\end{enumerate}
\end{theorem}

For simplicity, we start with our Monte Carlo version of our algorithm, and later we discuss our Las Vegas algorithm.

\subsection{Our Reinforced Tree-Like Data Structure: A Monte Carlo Version}
\label{Algorithm:Monte_Carlo}

In the remainder of this section, we demonstrate how to construct a data structure $\mathcal{D}$ that employs a tree structure to efficiently answer \MF queries for an input graph $G=(V,E)$. For a more conceptual grasp of the framework, we recommend the reader to become familiar with the original approximation algorithm designed for the edge-capacities case~\cite{AKT20}.

\paragraph*{Preprocessing}
At a high level, the basic framework involves recursion. We begin by partitioning an instance $V'\subseteq V$ using an \emph{expansion} operation into a few sets $S_i \subseteq V'$, on which we will apply the operation recursively until they have a size of 1. 
This partition will (almost) satisfy property (**) that was discussed in Section~\ref{it:reinforced_starproperty}.
To help us keep track of these partitions, we maintain a recursion-tree $T$. Each node $t_{V'}$ in the tree corresponds to an expansion operation, and stores $V'$ along with some auxiliary information, such as a mapping from each node $v \in V'$ to a cut $(S_{p,f(v)}, C_{p,f(v)}, S_{f(v),p})$, as well as to other cuts computed directly for it.
To perform a query on a pair $u,v\in V$ we first check whether $\MF(u,v)$ was computed in one of the direct queries throughout the algorithm run. If not, we immediately gain information about the structure of the minimum vertex $uv$-cut. Then, we search for the recursion-node in $T$ that separated $u$ and $v$, i.e. the last $V'$ that contains both $u$ and $v$, and output the smaller of (at most) two cuts stored in that node.

We will now shift our focus to the recursion tree and later prove that its depth is bounded by $O(\log{n})$. At each depth of the tree, expansion operations are executed on disjoint subsets $V'$. This is crucial for bounding the total number of \MF queries involved in the expansion operations at that depth, which we will show is at most $\tO(\alpha(V))$.
More precisely, we run this algorithm $O(\log n)$ times on $V'=V$, thus constructing $O(\log n)$ recursion trees and demonstrate that with high probability, for every $u,v$ query the correct answer can be obtained from at least one of the recursion trees.

The following is the expansion operation on a subset $V' \subseteq V$. If $V'=V$ (i.e., in the beginning) then repeat $4\log n$ times.
\begin{enumerate}
\item\label{step:alg:1} let $V'_H:=\{u\in V':\alpha(u)\geq|V'|/20\log_{4/3} n\}$, and compute directly $\MF(u,v)$ for every pair of nodes $u,v\in V'_H$.
\item\label{step:alg:2} pick a pivot node $p\in V'$ uniformly at random. 

\item\label{step:alg:3} for every node $u\in V'\setminus \{p\}$ perform a query $\MF(p,u)$ to get a vertex-cut $(S_{p,u}, C_{p,u}, S_{u,p})$ where $p \in S_{p,u}$, $u \in S_{u,p}$, and the vertex-cut separating $u$ and $v$ (i.e. a $uv$-separator) is $C_{p,u}$.
Then, compute the intersection with $V'$ of the side of $u$ in the cut, that is $S_{u,p}\cap V'$, and denote this set by $S'_{u,p}$, and similarly define $S'_{p,u}$ and $C'_{p,u}$.

\item\label{step:alg:4}
the following is a \emph{reassignment} process, which is designed for our case of vertex-cuts. We aim to satisfy three main criteria when reassigning nodes to cuts.

\begin{enumerate}
\item we may only assign a node $v$ to a cut $(S_{p,u}, C_{p,u}, S_{u,p})$ with the same value as $\val(C_{p,v})$, which is an optimal cut for $v$. This is necessary to eventually output optimal cuts.

\item we prioritize assigning $v$ to an optimal $pu$-cut $(S_{p,u}, C_{p,u}, S_{u,p})$ that separates $v$ from $p$ and has a \emph{small cardinality} $|S'_{u,p}|$. This ensures that the sets significantly decrease in size with each recursive step and thus upper bounds the total recursion depth by $O(\log{n})$.

\item we can only assign two nodes $u,v$ to different sets if we have evidence for doing so in the form of a (likely optimal) vertex-cut $C$ that separates one but not the other from $p$, and thus separates them.

\end{enumerate}
Satisfying all of the above criteria at the same time requires the following process. We define a reassignment function $f:V'\rightarrow V'\cup \{\bot\}$ such that for every node $u\in V'\setminus \{p\}$ with cut $(S_{p,u}, C_{p,u}, S_{u,p})$, we reassign $u$ to $v$, denoted $f(u)=v$ with the cut $(S_{p,f(u)}, C_{p,f(u)}, S_{f(u),p})$ as follows.
Denote by $V_{small}, V'_{small}$ two initially identical sets, each containing all nodes $u$ such that $\card{S'_{u,p}}\leq n'/2$, and denote by $V_{big},V'_{big}$ two sets, both are initially equal to $V'\setminus V_{small}$.

Sort $V_{small}$ by $\MF(p,u)$, and for all $u\in V_{small}$ from low $\MF(p,u)$ to high and for every node $v\in S'_{u,p}\cap V'_{small}$, set $f(v)=u$ and then remove $v$ from $V'_{small}$. 
Furthermore, for every node $v\in S'_{u,p}\cap V'_{big}$, if $\MF(p,u) = \MF(p,v)$ then set $f(v)=u$ and then remove $v$ from $V'_{big}$.
Finally, set $f(v)=\bot$ for every node $v$ for which $f$ was not assigned a value (including $p$). 

To get the partition, let $IM(f)$ be the image of $f$ (excluding $\bot$) and for each ${i\in IM(f)}$ let $f^{-1}(i)$ be the set of all nodes $u$ that were reassigned by $f$ to the cut $(S_{p,f(i)}, C_{p,f(i)}, S_{f(i),p})$. 
Notice that the nodes in $V'_{big}$, which include $p$, were not assigned to any set. Thus, we get the partition of $V'$ into $V'_{big}$ and each set in $\{f^{-1}(i)\}_{i\in IM(f)}$. The latter sets may satisfy property (**), but $V'_{big}$ may not; in particular, (**).\ref{property:3} could be violated because $V'_{big}$ does not correspond to a minimum cut, and therefore it will be handled separately shortly.

\item\label{step:alg:5} if $|V_{small}|<n/4$ then $p$ is a failed pivot. In this situation, re-start the expansion operation at step~\ref{step:alg:1} and continue choosing new pivots until $|V_{small}|\geq n/4$, or until the number of repetitions is more than $4\log n$, in which case output ``failure''. We show that the latter occurs with low probability.

\item\label{step:alg:6} for every $u\in V'\setminus \{p\}$ and $q\in C_{p,f(u)}$, query directly $\MF(q,u)$. 

\item\label{step:alg:7} finally, we recursively compute the expansion operation on each of the sets of the partition, while their sizes are bigger than $1$.
Let us describe what we store at the recursion node $t_{V'}$ corresponding to the just-completed expansion operation on $V'$ with pivot $p$.
We store $(|V'|-1)+\alpha(V')$ cuts in $t_{V'}$, including: 
\begin{enumerate}
\item for each $u,v\in V'_H$ we store $C_{u,v}$.
\item for each node $u_{big} \in V'_{big}\setminus \{p\}$ we store the cut $(S_{p,u_{big}}, C_{p,u_{big}}, S_{u_{big},p})$.
\item for each node $u$ in one of the other sets of the partition $\{f^{-1}(i)\}_{i\in IM(f)}$ we store the cut it was reassigned to $(S_{p,f(u)}, C_{p,f(u)}, S_{f(u),p})$.
\item for each $u\in V'$ and $q\in C_{p,f(u)}$ we store $C_{q,u}$. 
\end{enumerate}
We also keep an array of pointers from each node to its corresponding cuts and also their values, and we call this array $A$.
We also store, for each node of $V'$, the name of the set in the partition that it belongs to, in an array $B$.
\end{enumerate}

\paragraph*{Queries}
To answer a query for a pair of nodes $u$ and $v$, we first check if $\MF(u,v)$ was computed directly for them. If so, we output the corresponding separator. Otherwise, we do the following for each recursion tree $T$ and pick the answer with minimal value: we go to the recursion depth that separated $u$ and $v$, which corresponds to some node $t_{V'}$ in $T$, and output a pointer to one of the two corresponding cuts $(S_{p,f(u)}, C_{p,f(u)}, S_{f(u),p})$ or $(S_{p,f(v)}, C_{p,f(v)}, S_{f(v),p})$ -- the one with smaller capacity among the (at most) two that separate $u$ and $v$. We prove that at least one of the two cuts separates them.
To find out which recursive node separates $u$ and $v$, we start from the root and continue down, guided by the array $B$, to the nodes that contain both $u$ and $v$ until we reach a node that separates them. The query time depends on the depth of the recursion, which we will show later to be logarithmic for all trees.
\paragraph*{Correctness}
The claim below proves that the cut returned for each pair $u,v$ is optimal with high probability. The main idea is that either both $\alpha(u)$ and $\alpha(v)$ are high with respect to $n':=|V'|$ in the instance $V'$ that separates them, in which case the algorithm computes $\MF(u,v)$ directly, or at least one of $\alpha(u)$ and $\alpha(v)$ is small, in which case their separator is small, and $p$ is unlikely to be in it. Then, assume without loss of generality that $\MF(p,u)\leq \MF(p,v)$. Consequently, either $C_{p,u}$ is optimal for $\MF(u,v)$, or $v\in C_{p,u}$, and we computed $\MF(u,v)$ directly. More formally,

\begin{lemma}\label{alg1:correctness}
With high probability, the separator returned for any pair of nodes $u,v$ is a minimum $uv$-separator.
\end{lemma}

\begin{proof}
Let $T$ be a recursion tree, and let $V' \subseteq V$ be the instance where $u$ and $v$ belong to $V'$ but were sent to different sets during the expansion operation. We will use the following claim to make the assumption that the pivot $p$ in $V'$ satisfies $p\notin C_{u,v}$.
\begin{claim}\label{claim:p}
With high probability, for at least one recursion tree $T$ either $\MF(u,v)$ is computed in step~\ref{step:alg:1}, or $p\notin C_{u,v}$ for the instance separating $u$ and $v$.
\end{claim}
\begin{proof}
For a node $r\in V$, let $V'_r$ be an instance in $T$ that contains $r$, and is the first instance satisfying $\alpha(r) \geq |V'_r|/(20\log_{4/3} n)$. Let $V'_{u,v}$ be $V'_{w}$ for $w=\arg_{u,v}\min\{\alpha(u),\alpha(v)\}$, and let $\mathcal{X}_{u,v}$ be the event that $p\notin C_{u,v}$ for all instances $V'$ along the path from the root to the parent of $V'_{u,v}$ in $T$. 
Note that $\mathcal{X}_{u,v}$ trivially occurs for $V'_{u,v}=V$.
The probability of $\mathcal{X}_{u,v}$ is at least $(1-1/(10\log_{4/3} n))^{\log_{4/3}n}\geq 1/2$, as follows.
A uniformly random node in one such $V'$ does not belong to $C_{u,v}$ with probability $\frac{|V'|-|C_{u,v}|}{|V'|}\geq (1-1/20\log_{4/3} n)$, then recall that only successful pivots could be chosen, and by Corollary \ref{Corollary:Tournament} below, their number is at least $|V'|/2$. Lastly, we raise to the power of $\log_{4/3} n$, as we will show later that the depth of $T$ is at most this value.
Now, if $\mathcal{X}_{u,v}$ occurs in $T$ then either $u$ and $v$ both survived the expansions and belong to $V'_{u,v}$ and then they are queried directly in step~\ref{step:alg:1}, or $z=\arg_{u,v}\max\{\alpha(u),\alpha(v)\}$ was partitioned away in a higher instance $V'$ in $T$, in which case $p\notin C_{u,v}$ for $V'$.

The probability that $\mathcal{X}_{u,v}$ does not occur in any tree is at most $(1/2)^{4\log n}\leq 1/n^4$. By applying a union bound over all $\leq n^2$ node pairs, we can conclude that with high probability, $\mathcal{X}_{u,v}$ occurs in at least one of the recursion trees for every pair of nodes $u,v$. This completes the proof of Claim~\ref{claim:p}.
\end{proof}

Thus, henceforth we assume that
\begin{equation}\label{eq_p}
p\notin C_{u,v}.
\end{equation}

An additional assumption we make is that for any pair $u,v$ considered,
\begin{equation}\label{eq_cross}
u\notin C_{p,v}\wedge v\notin C_{p,u}.
\end{equation}
This can be justified since otherwise the minimum $uv$-cut is computed directly in step~\ref{step:alg:6} of the algorithm. 
Our final assumption is that without loss of generality,
\begin{equation}\label{eq_u_v}
\MF(p,u)\leq \MF(p,v).
\end{equation}
Next, there are a few cases, depending on whether any of $u,v$ is in $V'_{big}$ or not.

\begin{enumerate}
\item
First, consider the case when none of $u,v$ are in $V'_{big}$. There are two sub-cases, depending on whether the values of the two cuts are equal or not.
\begin{enumerate}
\item
If
$$
\MF(p,u)<\MF(p,v),
$$
then 
$$\MF(u,v)=\MF(p,u).$$ 
Observe that $v\notin S'_{f(u),p}$ because otherwise, we would have a smaller cut than the minimum cut found for $u$. Together with~\ref{eq_cross} we conclude that $v\in S'_{p,f(u)}$, and so $\MF(u,v)\leq\MF(p,u)$.
Furthermore, if it was the case that 
$\MFV(u,v)<\MFV(p,u)$
then by~\ref{eq_p} it would be that $p\in S'_{u,v}$ or $p\in S'_{v,u}$ which would refute the minimum cut value found for the pairs $p,v$ or $p,u$, respectively.
Indeed, $(S_{p,f(u)}, C_{p,f(u)}, S_{f(u),p})$ is the cut returned.
\item
Otherwise,
$$
\MFV(p,u) = \MFV(p,v).
$$
First, we see that $v\notin S'_{f(u),p}$. Otherwise, as $u,v\in V_{small}$ they must have been sent to the same recursion instance, contradicting our assumption on the expansion operation on $V'$. Like before, it implies that $\MF(u,v)=\MF(p,u)$, and so for this pivot the algorithm can output $(S_{p,f(u)}, C_{p,f(u)}, S_{f(u),p})$ as the optimal cut for $u,v$, as required.
\end{enumerate}

\item Second, we consider the case when one of the nodes is in $V'_{big}$ and the other is in $V'_{small}$, and we separate into three sub-cases according to their Max-Flow value to $p$. 
More specifically, let $u_{big}\in V'_{big}$ and $u\notin V'_{big}$ be nodes where the minimum cut $(S_{p,u_{big}}, C_{p,u_{big}}, S_{u_{big},p})$ is the cut corresponding to $u_{big}$.

\begin{enumerate}
\item 
If
$$
\MF(p,u_{big})>\MF(p,f(u)),
$$ 
then $u_{big}\notin S'_{f(u),p}$ as otherwise $u_{big}$ would have been assigned to $f(u)$, as $f(u)$ is in $V_{small}$.
Together with~\ref{eq_cross}, we have that 
$u_{big}\in S'_{p,f(u)}$.
We conclude that $(S_{p,f(u)},C_{p,f(u)}S_{f(u),p})$ is a minimum $u_{big}v$-cut.
\item 
If
$$
\MF(p,u_{big})=\MF(p,f(u)),
$$
then it must be that
$u_{big}\notin S'_{f(u),p},$
since otherwise, as $u_{big}\in V'_{big}$ and when the algorithm examined $f(u)$ it did not set $f(u_{big}):= f(u)$, it must have been the case for a node $x$ that was either $f(u)$ or before $f(u)$ in the order (i.e. such that $\val(C_{p,x})\leq \val(C_{p,f(u)})$) that $u_{big}$ was tested for the first time, with $\val(C_{p,x})> \val(C_{p,u_{big}})$, and so 
$\val(C_{p,f(u)})> \val(C_{p,u_{big}})$.
However, by our assumption it holds that $\val(C_{p,f(u)})= \val(C_{p,u_{big}})$, in contradiction.
Together with~\ref{eq_cross} we get $u_{big}\in S'_{p,f(v)}$, and like before, the returned cut $(S_{p,f(u)}, C_{p,f(u)}, S_{f(u),p})$ is optimal for $u_{big},u$.
\item
Finally, if 
$$
\MF(p,u_{big}) < \MF(p,f(u)),
$$ 
then $(S_{p,u_{big}}, C_{p,u_{big}}, S_{u_{big},p})$ separates $u_{big}$ and $u$, and with~\ref{eq_cross} and~\ref{eq_p} providing an optimal minimum cut, concluding the claim.
\end{enumerate}
\end{enumerate}
\end{proof}

\paragraph*{Running time}
Here, we prove the upper bounds on the preprocessing time, by proving that with high probability, the algorithm terminates after $\tO(\alpha(V))$ calls to \MF and $\tO(\alpha(V)n)$ time outside of these calls.

At a high level, our aim is to bound the size of each of the sets in the partition in an expansion operation by $3|V'|/4$, which implies a bound of $O(\log n)$ on the height of the recursion tree. This bound is immediate for the sets $\{f^{-1}(i)\}_i$ since they are subsets of cuts $S_{u,p}$ of nodes $u$ in $V_{small}$ and, by definition, satisfy that $|S'_{u,p}| \leq n'/2$. Therefore, we only need to worry about $V'_{big}$.
However, any node $u$ that is initially in $V_{small}$ will end up reassigned to one of the sets $\{f^{-1}(i)\}_i$ and not to $V'_{big}$. Thus, it suffices to argue that $|V_{small}|\geq \card{V'}/4$. To this end, we cite a lemma from~\cite{AKT20} which shows that for a randomly chosen $p$, at least $1/4$ of the nodes $u \in V'$ satisfy that the side of $u$ is smaller than the side of $p$, and so they will end up in $V_{small}$. We provide a more formal proof of this below, starting with a general lemma about tournaments.

\begin{lemma}[Lemma $2.4$ in~\cite{AKT20}]\label{Lemma:Tournament}
Let $Y=(V_Y,E_Y)$ be a directed graph on $n$ nodes and $m$ edges that contains a tournament on $V_Y$. Then $Y$ contains at least $n/2$ nodes with an out-degree at least $n/4$.
\end{lemma}

The following is a corollary of Lemma~\ref{Lemma:Tournament} concerning cuts between every pair of nodes, and it is an adaptation of Corollary $2.5$ in~\cite{AKT20} to vertex-cuts.

\begin{corollary}\label{Corollary:Tournament}
Let $F=(V_F,E_F)$ be a graph where each pair of nodes $u,v\in V_F$ is associated with a cut $(S_{u,v}, C_{u,v}, S_{v,u})$ where $u\in S_{u,v}, v\in S_{v,u}$ (possibly more than one pair of nodes are associated with each cut), and let $V'_F\subseteq V_F$. 
Then there exist $\card{V'_F}/2$ nodes $p'$ in $V'_F$ such that at least $\card{V'_F}/4$ of the other nodes $w \in V'_F{\setminus} \{p'\}$ satisfy $\card{S_{p',w}\cap V'_F}\geq\card{ S_{w,p'} \cap V'_F }$. 
\end{corollary}

\begin{proof}
Let~$H_{F}(V'_F)$ denote the \emph{helper graph} of~$F$ on~$V'_F$, where there is a directed edge from $u\in V'_F$ to $v\in V'_F$ iff $\card{S_{u,v}\cap V'_F}>\card{S_{v,u}\cap V'_F}$.
By Lemma~\ref{Lemma:Tournament}, since $H_{F}(V'_F)$ contains a tournament on~$V'_F$,  Corollary~\ref{Corollary:Tournament} holds.
\end{proof}

Next, apply Corollary~\ref{Corollary:Tournament} on $G$, and let $H=H_G(V')$ be the helper graph of $G$ on $V'$ with the reassigned cuts.
As a result, with probability at least~$1/2$, the pivot $p$ is one of the nodes with out-degree at least $n'/4$.
In that case, when the algorithm partitions~$V'$, it must be that $\max_i \card{f^{-1}(i)}\leq n'/4$, and $\card{f^{-1}(\bot)}\leq 3n'/4$, that is, the largest set created is of size at most~$3n'/4$.
This is because, while the separated parts $S_{u,v}$ and $S_{v,u}$ for each $uv$-cut do not cover the entire node set as the case in~\cite{AKT20}, the condition $\card{S_{p,v}\cap V'_F}\geq\card{S_{v,p}\cap V'_F}$ for adding an edge in the tournament still implies that $\card{S_{f(v),p}\cap V'_F}\leq n'/2$, meaning that $v$ remains in $V_{small}$ after the reassignment process in our setting too.

After making $O(n\log n)$ successful pivot choices, the algorithm terminates with the total depth of the recursion tree being at most $\log_{4/3} n$. Note that in step~\ref{step:alg:5}, the algorithm verifies the choice of $p$ and only proceeds with a successful pivot choice. Therefore, it suffices to bound the running time of the algorithm given only successful pivot choices. The probability of having $4\log n$ consecutive failures in a single instance is bounded by $(1/2)^{4\log n} = 1/n^4$, and by applying the union bound over the total $\tO(n)$ instances in all recursion trees, the probability of having at least one instance that takes more than $4\log n$ attempts until a successful choice of $p$ is bounded by $1/n^2$.

We now turn our attention to bounding the total number of \MF queries. The total number of queries performed for pairs in step~\ref{step:alg:1}, i.e. where both nodes have degrees at least $n'/(20\log_{4/3} n)$,  in a single instance $V'$ is at most:
\begin{align*}
|V'_H|^2 &=
|\{u\in V' : \alpha(u)>|V'|/(20\log_{4/3} n) \}|^2 \leq
\sum_{u\in V'_H}\alpha(u)20\log_{4/3} n \leq 
\alpha(V'_H) \cdot 20\log_{4/3} n.
\end{align*}

The total number of queries performed in steps~\ref{step:alg:3} and~\ref{step:alg:6} in one instance $V'$ is at most
$$
|V'|+\sum_{u\in V'} \alpha(p,u)=|V'|+\alpha(V').
$$
Thus, the total number of queries performed in one depth $i$ is 
$$
\sum_{V'\text{ in depth }i}\left(\alpha(V'_H)\cdot 20\log_{4/3} n + 
|V'|+\alpha(V')\right)\leq \alpha(V)\cdot20\log_{4/3} n
+n+\alpha(V)\leq \tO(\alpha(V)),
$$
and summing over all depths, repetitions, and trees, we get a total number of \MF queries of $t_{pq}(m)=\tO(\alpha(V)).$
Note that most of the time spent outside of \MF calls is in steps~\ref{step:alg:3} and~\ref{step:alg:4}, and it is bounded by $O(n)$ per cut computed, thus a total of $t_{po}(m)=\tO(\alpha(V)n)$, as required.
\paragraph{Space usage}
We aim to bound the total space used by the data structure, focusing on the more demanding steps and beginning with item~\ref{DS:with} of Theorem~\ref{thm:algorithm:main:gen}.
\begin{itemize}
\item
To bound the space occupied by cuts in both step~\ref{step:alg:1} and step~\ref{step:alg:6}, 
we focus on a $[k,2k)$ range of separator sizes each time, incurring an $O(\log n)$ factor in the space bound.
On one hand, at most $(\alpha(V)/k)^2$ node pairs can have separators with size in $[k,2k)$, and on the other hand, as demonstrated earlier, at most $\tO(\alpha(V))$ reinforcement queries are made throughout the algorithm.
This yields a space bound of
$$ 
\tO\left(\min \left\lbrace 
(\alpha(V)/k)^2,\alpha(V) 
\right\rbrace \cdot 2k\right)\leq
\tO(\alpha(V)^{3/2}).$$
\item
For cuts in step~\ref{step:alg:3}, the occupied space is at most
$\tO\left(\sum_{V'} \alpha(V')+|V'|\right)\leq  \tO(\alpha(V)).$
\end{itemize}
Altogether, the bound is $\tO(\alpha(V)^{3/2})$, as required.

In order to obtain the data structure~\ref{DS:without} of Theorem~\ref{thm:algorithm:main:gen}, we do not keep the separators from steps~\ref{step:alg:1} and~\ref{step:alg:6} as they were only used as witnesses. 
The total space bound is thus at most $\tO(\alpha(V))$, as required.

\subsection{Our Reinforced Tree-Like Data Structure: A Las Vegas Version}\label{Algorithm:Las_Vegas}

In this section, we present a Las Vegas version of our algorithm, focusing on the differences from the Monte Carlo version presented above.

\paragraph*{Preprocessing}
At a high level, the basic framework is similar and involves recursion. However, here we have $O(\log n)$ stages, and in each stage we focus on \MF values that are in the range $[k,2k)$. In each stage the number of terminals, i.e., nodes $u$ such that $\alpha(u)\geq k$, is at most $\alpha(V)/k$. At the same time, the size of the optimal separators between node pairs in this stage is smaller than $2k$. 
Consequently, if such pairs are in instances of size $\geq 60k\log^3 n$ then many pivots can be chosen to separate them properly, which motivates making $\tO(k)$ repetitions of the recursion tree (we may use the terms repetitions and recursion trees interchangeably) and arguing that a high enough number of pivots will be chosen throughout the repetitions. 
This will guarantee that each pair $u,v$ has at least one repetition where the corresponding pivot does not belong to $C_{uv}$.
The resulting partitions after running the $\tO(k)$ repetitions will (almost) satisfy property (***) that was discussed in Section~\ref{it:reinforced_starproperty_}.
Indeed, we demonstrate that when the algorithm terimnates, with \textit{certainty} for every $u,v$ query the correct answer can be obtained from one of the recursion trees, or that we have queried the pair $u,v$ directly.

Let us focus on a range $[k,2k)$, and set $V_{\geq k}=\{u\in V: \alpha(u)\geq k\}$ as the set of terminals for the current stage.
The following is the expansion operation on a subset of terminals $V' \subseteq V_{\geq k}$.
If $V'=V_{\geq k}$ (i.e., in the beginning of the current stage) then repeat $60k\log^3 n$ times.
\begin{enumerate}
\item\label{step:algLV:1} pick a pivot node $p\in V'$ uniformly at random. 
\item\label{step:algLV:2} for every node $u\in V'\setminus \{p\}$ perform a query $\MF(p,u)$ to get a vertex-cut $(S_{p,u}, C_{p,u}, S_{u,p})$ where $p \in S_{p,u}$, $u \in S_{u,p}$, and the vertex-cut separating $u$ and $v$ (i.e. a $uv$-separator) is $C_{p,u}$.
Then, compute the intersection with $V'$ of the side of $u$ in the cut, that is $S_{u,p}\cap V'$, and denote this set by ${S'}_{u,p}$, and similarly define ${S'}_{p,u}$ and $C'_{p,u}$.
\item\label{step:algLV:3}
perform a reassignment process $f:V'\rightarrow V'\cup \{\bot\}$ on the nodes $V'$ and the corresponding cuts, similar to the Monte Carlo algorithm above, and set $n'=|V'|$.
\item\label{step:algLV:4} 
if either 
\begin{enumerate}
\item\label{case:succ_pivot_old} $|V_{small}|<n'/4$, or
\item\label{case:succ_pivot_new} $p$ is an overused pivot, i.e., $p$ was selected as a pivot in $8\log n$ previous repetitions, then
\end{enumerate}  
$p$ is a failed pivot. In this situation, re-start the expansion operation at step~\ref{step:alg:1} and continue choosing new pivots until both $|V_{small}|\geq n'/4$ and $p$ not an overused pivot.
We will show that with high probability, the number of pivot choices is $O(\log n)$.
\item\label{step:algLV:5} recursively compute the expansion operation on each set within the resulting partition whose size is larger than $60k\log^3 n$.
Let us describe what we store at the recursion node $t_{V'}$ corresponding to the just-completed expansion operation on $V'$ with pivot $p$.
We store $(|V'|-1)+\alpha(V')$ cuts in $t_{V'_{\geq k}}$, including: 
\begin{enumerate}
\item for each $u,v\in V'$ that we computed directly, we store $C_{u,v}$.
\item for each node $u_{big} \in V'_{big}\setminus \{p\}$ we store the cut $(S_{p,u_{big}}, C_{p,u_{big}}, S_{u_{big},p})$.
\item for each node $u$ in one of the other sets of the partition $\{f^{-1}(i)\}_{i\in IM(f)}$ we store the cut it was reassigned to $(S_{p,f(u)}, C_{p,f(u)}, S_{f(u),p})$.
\end{enumerate}
We also keep an array of pointers from each node to its corresponding cuts and also their values, and we call this array $A$.
We also store, for each node of $V'$, the name of the set in the partition that it belongs to, in an array $B$.
\item\label{step:algLV:6}
finally, for every pair of nodes $u,v$ such that 
\begin{enumerate}
\item\label{step:algLV:6:together}
$u\in S_{f(v),p}$ (i.e., $u$ and $v$ are not separated) in the final instances of at least $k$ repetitions, or
\item\label{step:algLV:6:in_separator}
$u\in C_{p,f(v)}$ or $v\in C_{p,f(u)}$ in at least $k$ repetitions,
\end{enumerate}
query $\MF(u,v)$ directly.
\end{enumerate}

\paragraph{Queries}
For data structure~\ref{DS:with} of Theorem~\ref{thm:algorithm:main:gen}, to answer a query for a pair of nodes $u$ and $v$, we operate similarly to the Monte Carlo algorithm; first check if $\MF(u,v)$ was computed directly. If so, we output the corresponding separator, and otherwise we pick the answer with minimal value in the recursion trees that separated $u$ and $v$ in all relevant stages.
The number of terminals in a $[k,2k)$ stage is at most $t_k=\min\{\alpha(V)/k,n\}$, and so the total time to preprocess all node pairs in all stages and repetitions is $\tO(\max_k\left\lbrace t_k^2\cdot k\right\rbrace )\leq \tO(\alpha(V)n)$, bringing the query time to $\tO(1)$, as required.
On the bounds in data structure~\ref{DS:without} of Theorem~\ref{thm:algorithm:main:gen} we elaborate in the \textit{Space usage} paragraph, later in this section.

\paragraph{Correctness}
Let $u,v$ be any pair of nodes. We claim that if the pair $u,v$ is not queried directly, then it must be separated by a pivot $p$ such that $p\notin C_{u,v}$, and then the correctness follows from the correctness proof in Section~\ref{Algorithm:Monte_Carlo}.
Next, we bound the amount of times a single pivot is chosen throughout the algorithm.

\begin{lemma}\label{lemma:pivots_bound}
With high probability, each pivot is chosen at most $8\log n$ times throughout all repetitions and depths.
\end{lemma}
\begin{proof}
Pivots are only chosen in instances of size greater than $60k\log^3 n$, and each recursion depth has always at most $60k\log^3 n$ instances containing $p$. As a result, the probability that the same pivot $p$ is chosen more than $8\log n$ times in the same recursion depth is at most $1/n^6$.
By the union bound over all nodes, repetitions, and depths, with high probability no node is selected as a pivot more than $8\log n$ times, as required.
\end{proof}

We say that a recursion tree becomes \emph{invalid} for a pair $u,v$ if $u$ and $v$ are separated in the tree ($u,v$ were assigned to different sets, i.e., $f(u)\neq f(v)$) but in a sub-optimal way, due to a pivot choice $p$ that belonged to $C_{u,v}$.
By Lemma~\ref{lemma:pivots_bound}, at every depth of the recursion at most $16k\log n$ recursion trees become invalid for the pair $u,v$. This is because the size of optimal separators in this stage is bounded by $2k$.
Since the height of the recursion tree is at most $\log_{4/3} n$ (see below), at the end of the recursion at least $60k\log^3 n-16\cdot (1/\log(4/3))\cdot k\log^2 n\geq 20k\log^3 n$ repetitions must not be invalid for $u,v$.
In those $20k\log^3 n$ remaining repetitions, one of the following three options must occur.
\begin{enumerate}[label=(\roman*)]
\item \label{case:B}
$u$ and $v$ are not separated in the final instances of at least $k$ repetitions.
\item\label{case:C}
$u\in C_{p,v}$ (or $v\in C_{p,u}$, symmetrically) in at least $k$ repetitions.
\item\label{case:proper_separation}
$u$ and $v$ are separated $\left(\text{i.e., }u\in S_{f(u),p}\text{ and }v\in S_{p,f(u)}\right)$ in at least one repetition by a pivot $p$ that does not belong to $C_{u,v}$.
\end{enumerate}

In cases~\ref{case:B} and~\ref{case:C} we must have queried the pair $u,v$ directly in steps~\ref{step:algLV:6:together} and~\ref{step:algLV:6:in_separator}, respectively. 
Otherwise, as the at most $2(k-1)$ repetitions where $u$ and $v$ fail to satisfy the requirements of step~\ref{step:algLV:6} can be identified, case~\ref{case:proper_separation} occurs and we immediately have the right answer in a recursion tree.

To bound the height of the recursion trees, observe that Corollary $2.5$ in~\cite{AKT20}, especially its adapted Corollary~\ref{Corollary:Tournament} mentioned above, is robust enough to accommodate the terminals, an aspect previously leveraged in those algorithms to handle the arbitrary node subsets resulting from the recursion.
Thus, the same argument from the Monte Carlo algorithm on the height of the recursion trees being at most $\log_{4/3} n$ follows here as well.
This concludes the correctness proof.

\paragraph{Running time}
A successful choice of a pivot $p$ in step~\ref{step:algLV:4} means that both (\ref{case:succ_pivot_old}) $|V_{small}|\geq n'/4$, and (\ref{case:succ_pivot_new}) $p$ is not an overused pivot.
The first case (\ref{case:succ_pivot_old}) is handled similarly to the Monte Carlo algorithm; the probability of $8\log n$ consecutive failures in a single instance is bounded by $(1/2)^{8\log n} = 1/n^8$, and by applying the union bound over the at most $\tO(n^2)$ instances in all repetitions, 
with probability at most $1/n^6$ at least one instance takes more than $8\log n$ attempts until a successful choice of $p$.
The second case (\ref{case:succ_pivot_new}) was handled in the correctness proof in Lemma~\ref{lemma:pivots_bound}. Again, by the union bound over the two cases, with high probability no instance takes more than $8\log n$ times to select a pivot, and no pivots are overused.
To bound the total number of \MF calls, consider the total number of terminal pairs in all repetitions (counting multiplicities) for which we could potentially directly call \MF.
\begin{itemize}
\item
for case~\ref{case:B}, this number is at most $\tO((\alpha(V)/k)/k \cdot k^2 \cdot k)\leq \tO(\alpha(V)k)$.
This is explained by having $(\alpha(V)/k)$ terminals in the $[k,2k)$ stage, considering only final (small) $\leq \tO(k)$ sized instances in this case, the at most $\tO(k^2)$ pairs of terminals within a small instance, and the $\tO(k)$ repetitions.
\item
for case~\ref{case:C}, this number is simply bounded by $\tO(\sum_{u\in V} \alpha(u)\cdot k)=\tO(\alpha(V)k)$.
\end{itemize}
Since we only call \MF directly for a node pair if it is required by at least $k$ repetitions of a certain type (either \ref{case:B}, or~\ref{case:C}), the total number of queries in both cases is bounded by $\tO(\alpha(V))$.
The basic recursion takes $\tO(k\cdot (\alpha(V)/k))\leq \tO(\alpha(V))$ calls to \MF, concluding the running time proof.

\paragraph{Space usage}
We aim to bound the total space used by the data structure, starting with item~\ref{DS:with} of Theorem~\ref{thm:algorithm:main:gen}.
\begin{itemize}
\item
For cuts in step~\ref{step:algLV:2}, an initial bound for the occupied space is at most
$$
\tO\left(\max_{k} \left\lbrace k\cdot\left(\min\{ \alpha(V)/k,n \}\cdot 2k\right)\right\rbrace \right)\leq  \tO(\alpha(V)n),
$$ 
taking into account that for a cut $\left(S_{p,f(u)},C_{p,f(u)},S_{f(u),p}\right)$ in this step and in a $[k,2k)$ stage, we only store the separator $C_{p,f(u)}$ if its size is smaller than $2k$, and the set $\{f^{-1}f(u)\}$ (i.e., the nodes assigned to the side opposite of $p$). 
However, this can be further improved as all node pairs $u,v\in V_{\geq k}$ are anyway queried directly for $k\geq \sqrt{\alpha(V)/(60\log^3 n)}$ in step~\ref{step:algLV:6:together} due to the halting condition in step~\ref{step:algLV:5}.
Consequently, we may perform those queries directly in a preprocessing step, performing the recursion only for stages where $k\leq \tO\left(\sqrt{\alpha(V)}\right)$, and so the space occupied in step~\ref{step:algLV:2} becomes $\tO\left(\alpha(V)^{3/2}\right)$, as required.
\item
For cuts in step~\ref{step:algLV:6}, recall that any query we make in this step is required by at least $k$ repetitions. Using the bound from the Monte Carlo algorithm, the occupied space here is at most
$$\tO\left(k\cdot\alpha(V)^{3/2}/k\right) \leq \tO\left(\alpha(V)^{3/2}\right).$$
\end{itemize}
Altogether, the bound is $\tO\left(\alpha(V)^{3/2}\right)$, as required.

In order to obtain data structure~\ref{DS:without} of Theorem~\ref{thm:algorithm:main:gen}, we do not store the separators from step~\ref{step:algLV:6} as they were only used as witnesses. However, we do maintain the list of $\tO(\alpha(V))$ node pairs we directly called \MF for in step~\ref{step:algLV:6}, together with the corresponding answers.
We further give up storing the separators from step~\ref{step:algLV:2}, notwithstanding the immediate consequence that for a cut $\left(S_{p,f(u)},C_{p,f(u)},S_{f(u),p}\right)$ and a node pair $u,v$ we cannot anymore immediately distinguish between the case where $u$ and $v$ are separated in the cut, i.e., $v\in S_{p,f(u)}$, and the case where $v\in C_{p,f(u)}$.
However, we may assume that a direct $\MF(u,v)$ call was not made in step~\ref{step:algLV:6:in_separator} (as we store the list of direct \MF queries), and so it must be that this $v\in C_{p,f(u)}$ case covered by step~\ref{step:algLV:6:in_separator}, and also the case where $v\in S_{f(u),p}$ in final instances covered by step~\ref{step:algLV:6:together}, happened in at most $k-1$ repetitions each.
As established in the correctness proof, at least $20k\log^3 n$ repetitions are not invalid for $u,v$ at the end of the recursion, and so at least $20k\log^3 n - 2(k-1)\geq 18k\log^3 n$ recursion trees must have the optimal answer for $u,v$.
To find the optimal answer, simply take the majority among the smallest $19k\log^3 n$ answers, i.e., in repetitions where $u$ and $v$ do not belong together in final instances, concluding a query time of $\tO(k)$ (as discussed before, node pairs in $[k,2k)$ stages with $k\geq \tilde{\Omega}\left(\sqrt{\alpha(V)}\right)$ can be answered faster, in $\tO(1)$ time).
The total space for data structure~\ref{DS:without} is thus at most $\tO\left(\alpha(V)\right)$, as required.

\section{Conditional Lower Bound for Undirected Graphs with Unit Node-Capacities}\label{sec:CLB}

Here, we show that the $n^{3-o(1)}$ conditional lower bound~\cite{AKT20_soda} for \APMF on undirected graphs with general node-capacities can be extended, using the source-sink isolating gadget presented in~\cite{HLSW23}, to prove a conditional lower bound for undirected graphs with unit node-capacities.

\begin{lemma}\label{lemma:CLB}
\TOV over vector sets of size $n$ and dimension $d$ can be reduced to \APMF on undirected graphs with $\Theta(nd)$ nodes, $\Theta(n^2d^2)$ edges, and unit node-capacities.
\end{lemma}

To see how Theorem~\ref{thm:CLB} follows from Lemma~\ref{lemma:CLB}, one can simply add isolated nodes and recall that the output is always of size $\Theta(n^2)$.

\begin{proof}[Proof of Lemma~\ref{lemma:CLB}]
We generally follow the reduction from \TOV presented in~\cite{AKT20_soda}, with minor modifications. As an intermediate construction in this reduction, a layered graph with undirected edges and a few directed edges was created. The following lemma describes the properties of this reduction, where \STMF is a restricted version of \APMF focusing only on $\MF(s,t)$ with $s\in S$ and $t\in T$, used here for convenience.

\begin{lemma}[From the proof of Lemma $4.2$ in~\cite{AKT20_soda}]\label{lemma:AKT_soda}
\TOV over vector sets of size $n$ and dimension $d$ can be reduced to $\STMF$ with $|S|=|T|=n$ in graphs $G$ with $\Theta(n\cdot d)$ nodes, $\Theta(n\cdot d)$ edges, and with the following properties.
\begin{itemize}
\item all nodes in $S\cup T$ have capacity $1$, and also $\sum_{u\in V(G)}c(u)=O(nd)$.
\item each edge $uv\in E(G)$ satisfies $c(u)\cdot c(v)\leq nd^2$.
\item the edges adjacent to $S$ are all directed away from $S$, and the edges adjacent to $T$ are all directed towards $T$. The rest of the edges in $G$ are undirected.
\end{itemize} 
\end{lemma}

It was shown that if an \STMF algorithm could distinguish between the following two cases, then the input \TOV instance can be solved. In the first case, all pairs of nodes $s\in S$ and $t\in T$ satisfy $\MF_{G}(s,t)=nd$, implying that there exists no orthogonal triple in the \TOV instance, i.e. it is a ``NO'' instance. In the second case, there exists at least one pair of nodes $s'\in S$ and $t'\in T$ such that $\MF_{G}(s',t')=nd-1$, indicating the presence of an orthogonal triple in the \TOV instance, i.e. it is a ``YES'' instance.
In the final construction, all edges must become undirected. To prevent flows from zigzagging through $S$, and similarly through $T$, all node-capacities are multiplied by $2n$, except for nodes in $S$ and $T$. It was then proved that the number of zigzagging flow paths in every $st$-flow is at most $|(S\cup T)\setminus \{s,t\}|\leq 2(n-1)$. Therefore, these paths can be disregarded after the node-capacities are multiplied by $2n$.

We will show that an alternative way to increasing the capacities in $G$ is to use the source-sink gadget. However, before doing so, we need to modify the graph so that it only has unit node-capacities. This can be easily done by replacing each node $u$ with $c(u)$ copies of itself, and by then adding a complete biclique between the new copies of $u$ and the new copies of $v$ for each edge $uv\in G$. The resulting graph is denoted by $G'$.

The number of nodes in $G'$ is clearly $\sum_{u\in V(G)}c(u)=O(nd)$, and the number of edges in $G'$ is $\sum_{uv\in E(G)}c(u)\cdot c(v)\leq O(nd\cdot nd^2)=O(n^2d^3)$ (in fact, a more careful analysis leads to a bound of $O(n^2d^2)$). Importantly, for every $s\in S$ and $t\in T$, every $st$-flow in $G$ naturally corresponds to an $st$-flow of the same value in $G'$, and vice versa; to transform an $st$-flow in $G'$ to an $st$-flow in $G$, simply contract nodes that correspond to the same node in $G$. Although the resulting flow may contain cycles, they can be disregarded.

Below is the source-sink isolating lemma, adapted to our notations.

\begin{lemma}[Lemma $3.2$ in~\cite{HLSW23}]\label{lemma:CleanupGadget}
Given an undirected graph $R$ and two disjoint groups of vertices $X,Y\subseteq V(R)$, there is a graph $Q$ with $V(Q)\cap V(R)=X\cup Y$ such that for any $x\in X,y\in Y$ with $(x,y)\notin E(R)$,
\[\MF_{R\cup Q}(x,y)=\MF_{R_{xy}}(x,y)+|X|+|Y|,\]
where $R_{xy}=R\setminus((X\cup Y)\setminus\{x,y\})$. Such a graph $P$ is called a source-sink isolating gadget and there is an algorithm that constructs $Q$ with $|V(Q)|=O(|X|+|Y|)$.
\end{lemma}

Constructing $Q$ can be easily done in linear time with respect to its size. We use Lemma~\ref{lemma:CleanupGadget}, taking $R=G'$, $X=S$ and $Y=T$. The lemma guarantees that $\MF_{G'\cup Q}(s,t)=\MF_{R_{st}}(s,t)+|S|+|T|$ for any $s\in S$ and $t\in T$. Therefore, we can determine whether there exists an orthogonal triple by checking whether $\MF_{G'}(s,t)=nd+2n$ for all $s\in S$ and $t\in T$, or if there exists $s'\in S$ and $t'\in T$ such that $\MF_{G'}(s',t')=nd-1+2n$.

The size of the graph remains asymptotically the same after adding the source-sink isolating gadget. Since \APMF is a more general problem than \STMF, we can conclude Lemma~\ref{lemma:CLB}.

\end{proof}

\section{Open Problems}
After our work, many open problems on the fine-grained complexity of \APMF remain. Here are a few:
\begin{enumerate}
\item Can we break the $O(n^4)$ bound for any hard setting? With our new non-reducibility results, (non-combinatorial) subquartic algorithms seem more attainable.
\item Can we establish a conditional lower bound of $n^2m^{1-o(1)}$  for all $m$, against combinatorial algorithms? Such a lower bound is not known even for directed graphs with general edge capacities.
\item Can we design an $\hO(m^2)$ time combinatorial algorithm for directed graphs with unit edge-capacities? Our adaptation of the framework in~\cite{AKT20} may hint its robustness.
\item Finally, the most significant open question after this work is: Can we design a subquartic nondeterministic algorithm for directed graphs with edge-capacities? Since our proof for the cut upper bounds in Lemma~\ref{lemma:nondetmaxflows} holds for this setting too, it is enough to develop a fast nondeterministic algorithm that proves flow lower bounds.
\end{enumerate}

\section{Acknowledgments}
We thank Julia Chuzhoy for valuable discussions and feedback, and Chandra Chekuri, Karthik Srikanta, and anonymous reviewers for helpful comments on this manuscript.
{\small
\ifprocs
\bibliographystyle{alpha}
\else
\bibliographystyle{alphaurlinit}
\fi
\bibliography{references}
}

\end{document}